\documentclass[lettersize,11pt]{article}

\usepackage{graphicx}
\usepackage[english]{babel}
\usepackage[latin1]{inputenc}
\usepackage[T1]{fontenc}
\usepackage{amsmath}
\usepackage{amscd}
\usepackage{amsfonts}
\usepackage{amssymb}
\usepackage{amsthm}
\usepackage{zref-perpage}
\usepackage{enumerate}
\usepackage{mathrsfs}
\usepackage{mathptmx}
\usepackage[margin=1.1in]{geometry}
\usepackage{wrapfig}
\usepackage{multicol}
\usepackage{mathptmx}

\makeatletter
\newif\if@restonecol
\makeatother

\usepackage[linesnumbered,noline,ruled]{algorithm2e}
\newcommand{\gconst}[1]{{\ensuremath{\mathbf{#1}}}}

\newcommand{\gqed}{}

\zmakeperpage{footnote}

\newtheorem{assumption}{Assumption}
\newtheorem{definition}{Definition}
\newtheorem{theorem}{Theorem}
\newtheorem{lemma}{Lemma}
\newtheorem{corollary}{Corollary}
\newtheorem{remark}{Remark}

\SetEndCharOfAlgoLine{}

\begin{document}

\begin{titlepage}
  \title{Self-Stabilizing Paxos\thanks{Regular Paper, Eligible for
      Best Student Paper Award. The paper should also be considered
      for the Brief Announcement.}}
  \author{ Peva
    Blanchard~\footnote{PhD student, LRI, University of Paris-Sud XI,
      Orsay, France.  Email: \texttt{blanchard@lri.fr}}\and Shlomi
    Dolev~\footnote{Dept. of Computer Science, Ben-Gurion Univ. of the
      Negev, Beer-Sheva, 84105, Israel.  Email:
      \texttt{dolev@cs.bgu.ac.il}.  Partially supported by Deutsche
      Telekom, Rita Altura Trust Chair in Computer Sciences, Intel,
      MFAT, MAGNET and Lynne and William Frankel Center for Computer
      Sciences.}\and Joffroy Beauquier~\footnote{LRI, University of
      Paris-Sud XI, Orsay, France.  Email: \texttt{jb@lri.fr}}\and
    Sylvie Delaët~\footnote{LRI, University of Paris-Sud XI, Orsay,
      France.  Email: \texttt{delaet@lri.fr}} }
\end{titlepage}

\date{}
\maketitle

\begin{abstract}
  We present the first self-stabilizing consensus and replicated state
  machine for asynchronous message passing systems.  The scheme does
  not require that all participants make a certain number of steps
  prior to reaching a practically infinite execution where the
  replicated state machine exhibits the desired behavior.  In other
  words, the system reaches a configuration from which it operates
  according to the specified requirements of the replicated
  state-machine, for a long enough execution regarding all practical
  considerations.
\end{abstract}

% \category{C.2.4}{Computer Communication Networks}{Distributed Systems}[Network operating systems]
% \category{D.4.5}{Operating Systems}{Reliability}[Fault tolerance]
% \terms{Algorithms, Theory, Reliability}
% \keywords{Distributed Algorithms, Consensus, Replicated State-Machine, Self-stabilization, Fault Tolerance, Paxos}

\thispagestyle{empty}

\newpage
\setcounter{page}{1}
%%%%%%%%%%%%%%%%%%%%%%%%%%%%%%%%%%%%%%%%%%%%%%%%%%%%%%%%%%%%%%%%%%%
\vspace*{-.2cm}
\section{Introduction}
\label{sec:introduction}
\vspace*{-.2cm}

One of the most influential results in distributed computing is Paxos \cite{lamportPartTimeParliament,lamport01paxos},
where repeated asynchronous consensus is used to replicate a state
machine using several physical machines.  The task is to use
asynchronous consensus, where safety is guaranteed, and liveness is
almost always achieved by using an unreliable failure detector\footnote{Note though that Paxos does not rely explicitly on a failure detector.}, to
implement an abstraction of very reliable state machine on top
of physical machines that can crash, though usually at least several
machines stay alive at any particular moment.  The extreme usefulness
of such an approach is proven daily by the usage of this technique, by
the very leading companies, to ensure their availability and
functionality.

Unfortunately Paxos is not self-stabilizing and therefore a single
transient fault may lead the system to stop functioning even when all
the cluster machines operate. One example is the corruption of the
time-stamp used to order operations in Paxos, where a single
corruption of the value of this counter to the maximal value will
cause the system to be blocked. In another scope, the occurrence of
transient fault with the same nature caused the Internet to be blocked
for a while \cite{Ros81}.  
Self-stabilization is a property that every
on-going system should have, as self-stabilizing systems automatically
recover from unanticipated states, i.e., states that have been reached due to
insufficient error detection in messages, changes of bit values in
memory \cite{DHSoftErrors}, and in fact any temporary violation in the
assumptions made for the system to operate correctly. The approach is
comprehensive, rather than addressing specific fault scenarios
(risking to miss a scenario that will later appear), the designer
considers every possible configuration of the system, where
configuration is a cartesian product of the possible values of the
variables. Then the designer has to prove that from every such a
configuration, the system converges to exhibit the desired behavior.

Self-stabilizing systems do not rely on the consistency of a
predefined initial configuration and the application of correct steps
thereafter. In contrast, self-stabilizing systems assume that the
consistency can be broken along the execution and need to
automatically recover thereafter. The designers assume an arbitrary
configuration and prove convergence, not because they would like the
system to be started in an arbitrary configuration, but because they
are aware that the specified initial configuration and the defined
steps consistency maybe temporarily broken, and would like the system
to regain consistency. Of course, such an approach does not guarantee
the convergence under a infinite stream of transient faults, which is
clearly impossible, but guarantees the system recovery after the last
transient fault. Therefore, although the system may lose safety
properties, the safety is automatically regained, leading to a safer
behavior than of non-stabilizing systems, namely, initially safe and
eventually safe \cite{DDPSafeEventuallySafe}. 

Self-stabilizing consensus and replicated state machine for shared
memory system appeared in \cite{dolev2010SelfStabConsensusSharedMem},
the case of message passing being left to future investigation.  One
approach to gain a self-stabilizing consensus and replicated state
machine in message passing is to implement the read-write registers
used in \cite{dolev2010SelfStabConsensusSharedMem}, using message
passing\footnote{A suggestion made by Eli Gafni.}. A self-stabilizing
implementation of such a single-writer multiple-reader register
appeared in \cite{DBLP:conf/sss/AlonADDPT11}. Unfortunately, the
implementation had to assume that the writer is active forever. Thus,
the implementation of self-stabilizing Paxos under the original
assumptions was left open.  In this paper we present the first
self-stabilizing Paxos in message passing systems.  One ingredient of
the self-stabilizing Paxos algorithm is a recent construction of a
self-stabilizing bounded time-stamp \cite{DBLP:conf/sss/AlonADDPT11}.
Note that the classical bounded time-stamp systems
\cite{Dolev:1997:BCT:249364.249372,Israeli:1993:BT:1080512.1080514}
{\em cannot} be started with an arbitrary set of label values as the
ordering is defined only for certain combination of labels (and
missing labels).  Such restricted combinations can be preserved by
non-stabilizing algorithms as long as transient faults do not occur.
Another bounded weak time-stamp \cite{Abraham:2003:ST:964895.964909}
was designed for particular shared memory self-stabilizing systems,
where participants have access to shared memory of others (while we
deal with message passing where many label values can be in messages
in transient). This bounded time-stamp allows a limited new-old
version inversion, and therefore does not guarantee the eventual
strict ordering of events, ordering that a replica state machine
should (eventually) promise.  Obtaining a self-stabilizing Paxos
requires to cope with many aspects including a way to  compose
self-stabilizing bounded time stamps, such that each time stamp is
governed by a distinct participant. In particular, the algorithm also needs
to cope with an arbitrary set of messages that are stored in the
system, as we demonstrate in the sequel. The paper starts with a
background and description of techniques and correctness in a
nutshell. Then we turn to a more formal and detailed description.

%%%%%%%%%%%%%%%%%%%%%%%%%%%%%%%%%%%%%%%%%%%%%%%%%%%%%%%%%%%%%%%%%%%
\vspace*{-.2cm}
\section{Self-Stabilizing Paxos Overview}
\label{sec:towards_practss}
\vspace*{-.2cm}

In this section, we define the Repeated Consensus Problem, show how it
can be used to implement a self-stabilizing replicated state machine
and give an overview of the Paxos Algorithm.  In addition, we give
arguments for the need of a self-stabilizing algorithm that would
solve the Repeated Consensus Problem.  Doing so, we investigate a new
kind of self-stabilizing behaviour, namely the \emph{practically
  self-stabilizing} behaviour. Here and everywhere, the semantical synonym for
  practically self-stabilizing is essentially self-stabilizing.

\paragraph{Repeated Consensus.}
\label{subsec:repeated_consensus}

The processors have to perform successive instances of consensus on
values proposed by some of them.  Every processor is assumed to have
an integer variable $s$, namely the step variable, that denotes the
current consensus instance it is involved in.  In each consensus
instance, processors decide on a value.  For example, in the context
of replicated state machines, the step variable denotes the current
step of the state machine, and at each step, a processor may decide to
apply a command to its copy of the state machine.  Processors may have
different views on what is the current step since some of them may
have progressed faster than others.  The  Repeated Consensus
Problem is defined by the following conditions: \emph{(Safety)} for
every step $s$, if two processors decide on values in step $s$, then
the two decided values must be equal, \emph{(Integrity)} for every
$s$, if a processor decides on some value, then this must have been
proposed in step $s$, \emph{(Liveness)} every non-crashed processor
decides in infinitely many steps.

\paragraph{Original Paxos.}
\label{subsec:original_paxos}

The original Paxos algorithm guarantees the safety and the integrity
property in an asynchronous complete network of processors
communicating by message-passing such that less than half of the
processors are prone to crash failures.  The algorithm uses unbounded
integers and also assumes that the system starts in a consistent
initial configuration.  To guarantee the liveness property, additional
assumptions must be made as discussed below.  The Paxos algorithm
defines three roles: \emph{proposer}, \emph{acceptor} and
\emph{learner}.  The proposer basically tries to impose a consensus
value for its current step\footnote{Note that there might be more than
  one proposer in each step.}. The acceptor accepts consensus values
according to some specific rules.  A value can be decided on for step
$s$ when a majority of acceptors have accepted it in step
$s$. Finally, the learner learns when some value has been accepted by
a majority of acceptors for some step and decides accordingly.  Here,
we assume that every processor is a learner and an acceptor, while
some processors can also be proposers.  Every proposer has its own
idea of what should the value be for step $s$.  For each step $s$, a
proposer executes one trial, or more, to impose some consensus
value. Thus, each processor maintains a Paxos tag, namely a couple
$(s~t)$ where $s$ denotes a step, i.e., a consensus instance, and $t$
a trial within this step. The Paxos algorithm assumes that all the
step variables and the trial variables are natural integers, hence
unbounded, and initially set to zero.  The Paxos tags are used to
timestamp the proposals emitted by the proposers or accepted by the
acceptors. To impose a proposal with tag $(s~t)$, a proposer must
first (phase 1) reads the most recent accepted proposal from a
majority of acceptors and try to impose its tag $(s~t)$ as the
greatest tag on this majority of acceptors; knowing that an acceptor
adopts the tag $(s~t)$ if it is strictly greater than its
own. Secondly (phase 2), the proposer tries to make a majority of
acceptors accept the previous read consensus value if it is not null,
or its own value otherwise; knowing that an acceptor accepts the
proposal if the tag $(s~t)$ is greater than or equal to its own. If
the proposer suceeds in these two phases, it decides on the proposal
and notifies the other processors.

The integrity property is guaranteed by the fact that a decided value
always comes from a proposer in the system. The difficulty lies in
proving that the safety property is ensured.  Roughly speaking, the
safety correctness is yielded by the claim that once a proposer has
succeeded to complete the second phase, the consensus value is not
changed afterwards for the corresponding step.  Ordering of events in
a common processor that answers two proposers yields the detailed
argument, and the existence of such a common processor stems from the
fact that any two majorities of acceptors always have non-empty
intersection.
The liveness property, however, is not guaranteed \cite{FLP85}. However, a close look at the behaviour of Paxos shows
that only the liveness property cannot be guaranteed and why it is so.
Indeed, since every proposer tries to produce a tag that is
 greater than the tags of a majority of  acceptor, two
such proposers may execute many trials for the same step without ever
succeeding to complete a phase two.  
Intuitively though,
it is clear that if, for any step, there is a single proposer in the system during a
long enough period of time, then the processors eventually decide in that step.

\paragraph{Self-Stabilizing Paxos.}
\label{subsec:practss}

As we pointed out in the previous section, the Paxos algorithm uses
unbounded integers to tag data.  In practice, however, every integer
handled by the processors is bounded by some constant
$2^{\mathfrak{b}}$ where $\mathfrak{b}$ is the integer memory size.
Yet, if every integer variable is initialized to a very low value, the
time needed for any such variable to reach the maximum value
$2^{\mathfrak{b}}$ is actually way larger than any reasonable system's
timescale.  For instance, counting from $0$ to $2^{64}$ by
incrementing every nanosecond takes roughly $500$ years to complete.
Such a long sequence is said to be practically infinite.  
This leads to the following important remark from which the current work stems.
\begin{remark}[Paxos and Bounded Integers]
  \label{rem:paxos_bint}
  Assuming that the integers are theoretically unbounded is reasonable
  only when it is ensured, in practice, that every step and trial
  variables are initially set to low values, compared to the maximum
  value. In particular, any initialized execution of the Paxos
  algorithm with bounded integers is valid as long as the counters are
  not exhausted.
\end{remark}

In the context of self-stabilization, however, a transient fault may
produce fake decision messages in the communication channels, or make
an acceptor accepting a consensus value that was not proposed. Such
transient faults only break the  Repeated Consensus conditions
punctually and nothing can be done except waiting.  However, a
transient fault may also corrupt the Paxos step and trial variables in
the processors memory or in the communication channels, and set them
to a value close to the maximum value $2^{\mathfrak{b}}$.  This leads
to an infinite suffix of execution in which the  Repeated Consensus
conditions are never jointly satisfied. This issue is much more worrying
than punctual breakings of the  Repeated Consensus specifications.
Intuitively though, if one can manage to get every integer variable
(step and trials) to be reset to low values at some point in
time, then there is consequently a finite execution (ending with step
or trial variables reaching the maximum value $2^{\mathfrak{b}}$)
during which the system behaves like an initialized original Paxos
execution that satisfies the  Repeated Consensus Problem
conditions\footnote{Modulo the unavoidable punctual breakings due to,
  e.g., fake decision messages.}.  Since we use bounded integers, we
cannot prove the safe execution to be infinite, but we can prove that this
safe execution is as long as counting from $0$ to $2^{\mathfrak{b}}$, which
is as long as the length of an initialized and safe execution assumed in the
original Paxos prior to exhausting the counters (cd Remark~\ref{rem:paxos_bint}).  This is what we call
a \emph{practically self-stabilizing} behaviour.

% The possibility of transient failures leads to many other issues. For example, some
% processors may, for instance, enter a loop waiting for messages that
% have not been sent, due to some corrupted program counter or some
% corrupted received message. In such a case, these processors are
% likely to be considered crashed, 

Replicated state machines have to perform steps that are commonly
decided. In the original Paxos, decisions on steps at a processor may
be learned out of order, but eventually every decision arrives, and
therefore the processor can also perform locally the agreed upon steps in a
sequence.  To avoid gaps in the sequence of agreed upon steps, it is
possible to use the Generalized Paxos approach
\cite{Lamport05generalizedconsensus}, where decisions are made on the
entire known sequence of steps, together with the new proposed step.
In the case of self-stabilization and when there is a need for
(eventual) identical steps execution by each participant, rather than
merely only a simulation of a global robust virtual state machine, the
decision subject is histories rather than the last state and next
step.  We mainly focus on the repeated
consensus version that can decide on the last state of the replicated
state machine and the next step, and then detail in Appendix
\ref{app:generalized_paxos} the very few modifications needed to
obtain the Generalized Self-Stabilizing Paxos.

The repeated consensus on both the current state and the step
requires, on the one hand, more communication, but on the other hand, addresses
a long standing technicality of memory garbage collection from the
array used to accumulate decided steps, as the decided
last current step encapsulates all step history prior to its
execution. The proposers always proposes a step that immediately
follows the last decided state it knows, and does not propose a new
step before deciding, or learning about a decision on this or a
subsequent state and step.

%%%%%%%%%%%%%%%%%%%%%%%%%%%%%%%%%%%%%%%%%%%%%%%%%%%%%%%%%%%%%%%%%%%
\vspace*{-.2cm}
\section{System Settings}
\label{sec:model}
\vspace*{-.2cm}

All the basic notions we use (state, configuration, execution,
asynchrony, \dots) can be found in, e.g.,
\cite{dolev2000self,Lynch:1996:DA:525656}.  Here, the model we work with is
given by a system of $\gconst{n}$ asynchronous processors in a
complete communication network. Each communication channel between two
processors is a bidirectional asynchronous communication channel of
finite capacity $\gconst{C}$ \cite{DBLP:conf/sss/DolevHSS12}.  Every
processor has a unique identifier and the set $\Pi$ of identifiers is
totally ordered. If $\alpha$ and $\beta$ are two processor
identifiers, the couple $(\alpha,\beta)$ denotes the communication
channel between $\alpha$ and $\beta$.  A configuration is the vector
of states of every processor and communication channel. If $\gamma$ is
a configuration of the system, we note $\gamma(\alpha)$
(resp. $\gamma(\alpha,\beta)$) for the state of the processor $\alpha$
(resp. the communication channel $(\alpha,\beta)$) in the
configuration $\gamma$.  We informally\footnote{For a formal
  definition, refer to, e.g., \cite{dolev2000self,Lynch:1996:DA:525656}.} define
an event as the sending or reception of a message at a processor or as
a local state transition at a processor. Given a configuration, an
event induces a transition to a new configuration. An execution is
denoted by a sequence of configurations $(\gamma_k)_{0 \leq k < T}$,
$T \in \mathbb{N}\cup\{+\infty\}$ related by such
transitions\footnote{For sake of simplicity, the events and the
  transitions are omitted.}. A local execution at processor $\lambda$
is the sequence of states obtained as the projection of an execution
on $\lambda$.  The initial configuration of every execution is
arbitrary and at most $\gconst{f}$ processors are prone to crash
failures.  A quorum is any set of at least $\gconst{n}-\gconst{f}$
processors.  For any execution $E$, we note $Live(E)$ the set of
processors that do not crash during $E$, and we note $Crashed(E)$ the
complement of $Live(E)$. We make the following resilience assumption.
\begin{assumption}[Resilience]
  \label{hyp:resilience}
  The maximum number of crash failures $\gconst{f}$ satisfies
  $\gconst{n} \geq 2\cdot\gconst{f}+1$. Thus, there always exists a responding majority quorum
  and any two quorums have a non-empty intersection.
\end{assumption}
We also use the ``happened-before'' strict partial order introduced by
Lamport \cite{DBLP:journals/cacm/Lamport78}.  In our case, we note $e
\leadsto f$ and we say that $e$ happens before $f$, or $f$ happens
after\footnote{Note that the sentences ``$f$ happens after $e$'' and
  ``$e$ does not happen before $f$'' are not equivalent.}  $e$.
In addition, every processor has access to a read-only boolean
variable $\Theta_{\alpha}$, e.g., from an unreliable failure detector
(Section~\ref{sec:liveness}) that satisfies the following condition.
\begin{assumption}[Module $\Theta$]
  \label{hyp:module_theta}
  For every infinite execution $E_{\infty} = (\gamma_k)_{k \in
    \mathbb{N}}$, there is a non-empty set $\mathcal{P}(E_{\infty})$,
  namely the \emph{proposers} in $E_{\infty}$, of processors in
  $Live(E_{\infty})$, such that, for every processor $\lambda$ in $\mathcal{P}(E_{\infty})$, the value of $\Theta_{\lambda}$ is always \gconst{true}, 
  and for every live processor $\mu$ not in $\mathcal{P}(E_{\infty})$, the value of $\Theta_{\mu}$ is eventually always \gconst{false}.
  % \begin{equation}
\end{assumption}
Note that this module is extremely weak in the sense that it simply
guarantees that at least one proposer is active. This proposer is not
required to be unique in order for our algorithm to stabilize. A
unique proposer is required only for the liveness of Paxos.

% In addition, there is an external module $\Theta$ such that, for every
% configuration in any execution, for every processor $\alpha$,
% $\Theta_{\alpha}$ is a boolean value.  We assume that, in any
% execution, there is at least one processor $\alpha$ such that the
% boolean $\Theta_{\alpha}$ is infinitely often equal to
% $\gconst{true}$.  Finally, every processor can send and receive
% messages to any other processor, and we assume that no message is
% lost, i.e., any message sent is eventually received.

%%%%%%%%%%%%%%%%%%%%%%%%%%%%%%%%%%%%%%%%%%%%%%%%%%%%%%%%%%%%%%%%%%%
\vspace*{-.2cm}
\section{Tag System Overview}
\label{sec:tag_system}
\vspace*{-.2cm}

This section presents the tag system used in our algorithm. For
didactic reasons, we first describe a simpler tag system that works
when there is a single proposer, before adapting it to the case of
multiple proposers. Formal definitions of bounded integers, labels and
tags are given in Appendix~\ref{app:tag_system_definitions}.

\paragraph{Single Proposer.}
\label{subsec:tag_system_single_proposer}

We start by looking at Paxos tags $(s~t)$ where the step $s$ and trial
$t$ variables are integers bounded by a large constant
$2^{\mathfrak{b}}$. Assume, for now, that there is a single proposer in
the system, and let's focus on its tag.  The goal of this proposer is
to succeed in imposing a consensus value for every step ranging from
$0$ to $2^{\mathfrak{b}}$, or at least from a low step value to a very
high step value.  The proposer can do a step increment, $(s~t)
\leftarrow (s+1~0)$, or a trial increment within the same step, $(s~t)
\leftarrow (s~t+1)$.  To impose some value in step $s$, it must reach
a trial $t$ such that the tag $(s~t)$ is lexicographically greater
than every other processor tags in a majority of acceptors.

With an arbitrary initial configuration, some processors may have tags
with step or trial value set to the maximum $2^{\mathfrak{b}}$, thus
the proposer will not be able to produce a greater tag. We thus define
a tag as a triple $(l~s~t)$ where $s$ and $t$ are the step and trial
fields, and $l$ a label, which is not an integer but whose type is
explicited below.  We simply assume that it is possible to increment a
label, and that two labels are comparable.  The proposer can increment
its trial variable, or increment its step variable and reset the trial
variable, or increment the label and reset both the step and the trial
variable. Now, if the proposer manages to produce a label that is
greater than every label of the acceptors, then it will succeed in a
practically infinite number of steps that mimicks the behaviour of the
original Paxos tags. To do so, whenever the proposer notices an
acceptor label which is not less than or equal to the proposer current
label (such an acceptor label is said to cancel the proposer label),
it records it in a history of canceling labels and produces a label greater than
every label in its history.

Obviously, the label type cannot be an integer. Actually, it is
sufficient to have some finite set of labels along with a comparison
operator and a function that takes any finite (bounded by some
constant) subset of labels and produces a label that is greater than
every label in this subset.  Such a device is called a finite labeling
scheme.  An implementation of such a finite labeling scheme was
suggested in \cite{DBLP:conf/sss/AlonADDPT11}, and is formally
presented in the Appendix~\ref{app:construction_labeling_scheme}.
Roughly saying, a label is a fixed length vector of integers from a
bounded domain in which the first integer is called sting and the
others are called antistings. A label $l_1$ is greater than a label
$l_2$, noted $l_1 \prec l_2$, if the sting of $l_1$ does not appear in
the antistings of $l_2$ but not vice versa. Given a finite set of
labels $l_1$,\dots,$l_r$, we can build a greater label $l$ by choosing
a sting not present in the antistings of the $l_i$, and choosing the
stings of the $l_i$ as antistings in $l$.  It is important to note
that the comparison relation between labels \emph{cannot be an order} since transitivity
does not hold.

%\subsubsection{Multiple Proposers}
\paragraph{Multiple Proposers.}
\label{subsec:tag_system_multiple_proposers}

In the case of multiple proposers, the situation is a bit more
complicated.  Indeed, in the previous case, the single proposer is the
only processor to produce labels, and thus it manages to produce a
label greater than every acceptor label once it has collected enough
information in its canceling label history.  If multiple proposers were also
producing labels, none of them would be ensured to produce a label that
every other proposer will use. Indeed, the first proposer can produce
a label $l_1$, and then a second proposer produces a label $l_2$ such
that $l_1 \prec l_2$.  The first proposer then sees that the label
$l_2$ cancels its label and it produces a label $l_3$ such that $l_2 \prec
l_3$, and so on.

To avoid such interferences between the proposers, we assume that the
set of proposer identifiers is totally ordered and we define a tag to
be a vector, say $a$, whose entries are indexed by the proposer identifiers.
Each entry $a[\mu]$ of the tag $a$ contains a tuple $(l~s~t~id~cl)$
where $l$ is a label, $s$ and $t$ are step and trial bounded integers,
$id$ is the identifier of the proposer that owns the tag, and $cl$ is
either a label that cancels $l$ or the null value\footnote{Which means
  that the label $l$ is not canceled.} denoted by $\bot$.  The
identifier of the proposer that owns the tag is included, so that two
proposers never share the same content in any entry of their
respective tags.  The canceling field tells the proposer whether the
corresponding label has been canceled by some label.  
% The role of
% cancellation is explained below.

Therefore, a proposer, say $\lambda$, has the possibility to use one
of the entries of its tag, say $a$, to specify the step and trial it is involved
in.  However, the entry used must be valid, i.e., the entry must
contain a null canceling field value along with step and trial values
strictly less than the maximum value $2^{\mathfrak{b}}$.  The entry
actually used by the proposer is determined by the lowest proposer
identifier, noted $\chi(a)$, such that the entry corresponding to $\chi(a)$ is valid.
The entry $a[\chi(a)]$ is referred to as the first valid entry in the tag.
If the first valid entry is located at the left of the entry indexed
by the proposer identifier, i.e., the identifier $\chi(a)$ is less than
the proposer identifier $\lambda$, then the proposer can increment the step and
trial values stored in the entry $a[\chi(a)]$, but it cannot increment the
label in the entry $a[\chi(a)]$.  The proposer can only increment the
label, and thus reset the corresponding step and trial variables,
stored in the entry indexed by its own identifier.
In addition, whenever the entry indexed by the proposer identifier
$\lambda$ becomes invalid, the proposer $\lambda$  produces a new
label in the entry $a[\lambda]$ and resets the integer variables to
zero and the canceling field to the null value $\bot$; this makes
$a[\lambda]$ a valid entry in the proposer tag.  The
important point is that, from a global point of view, the proposer
identified by $\lambda$ is the only proposer to introduce new labels in
the entries indexed by $\lambda$ in tags of the system.  Besides, this
also shows that any proposer $\lambda$ has to record in its canceling label history
only the canceling labels that are stored in the entry $\lambda$ of tags.

\begin{wrapfigure}{r}{0.5\textwidth}
  \centering
  \includegraphics[scale=0.45]{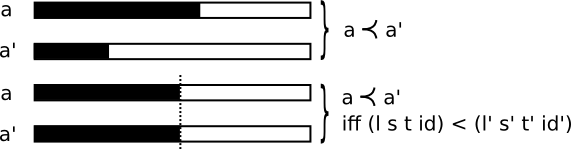}
  \caption{Comparison of tags - Invalid entries are darkened.}
  \label{fig:comparison_tags}
\end{wrapfigure}

A comparison relation is defined on tags so that every processor
(proposer or acceptor) always try to use the valid entry with the
lowest identifier. A tag $b_1$ is less than $b_2$, noted $b_1 \prec
b_2$, when either the first valid entry of $b_1$ is located at the
right of the first valid entry of $b_2$, or both first valid entries
are indexed by the same identifier $\mu$ and the tuple
$b_1[\mu].(l~s~t~id)$ is lexicographically less than the tuple
$b_2[\mu].(l~s~t~id)$.  We note $b_1 \simeq b_2$ when both tags share
the same first valid entry, and the corresponding contents are equal.
We note $b_1 \preccurlyeq b_2$ when $b_1 \prec b_2$ or $b_1 \simeq
b_2$.  If there is no valid entry in both tags, or if the labels are
not comparable, then the tags are not comparable.

%%%%%%%%%%%%%%%%%%%%%%%%%%%%%%%%%%%%%%%%%%%%%%%%%%%%%%%%%%%%%%%%%%%
\vspace*{-.2cm}
\section{The Algorithm}
\label{sec:algorithm}
\vspace*{-.2cm}

In this section, we describe the self-stabilizing Paxos algorithm.  We
first present the variables before giving an overview of the
algorithm.  The details of the algorithm and
the pseudo-code is given in Appendix~\ref{app:algorithm_details} and~\ref{app:algorithm_pseudo_code}.
In the sequel, we refer to the following datastructure.
\begin{definition}[Fifo History]
  A fifo history $H$ of size $d$ on a set $V$, is a vector of size $d$ of elements of $V$
  along with an operator $+$ defined as follows.
  Let $H = (v_1, \dots, v_d)$ and $v$ an element in $V$.
  If $v$ does not appear in $H$, then $H+v = (v,v_1,\dots, v_{d-1})$,
  otherwise $H+v = H$.
\end{definition}
We define the \emph{tag storage limit} $\gconst{K}$ and the
\emph{canceling label storage limit} $\gconst{K}^{cl}$ by
$\gconst{K} = \gconst{n} + \gconst{C}\frac{\gconst{n}(\gconst{n}-1)}{2}$
and $\gconst{K}^{cl} = (\gconst{n}+1)\gconst{K}$.
% We consider a tag system $(\mathfrak{b}, \Pi, \omega, \mathcal{L},
% \prec, d, \nu)$ such that $\Pi$ is the set of processor
% identifiers and the labeling scheme dimension is equal to $(\gconst{K}
% + 1)\gconst{K}^{cl}$.

% \subsection{Variables}
% \label{subsec:variables}

\paragraph{\bf Variables.}
The state of a processor $\alpha$ is defined by the following
variables: the \emph{processor tag} $a_{\alpha}$, the \emph{processor
  proposal} $p_{\alpha}$ (a consensus value), the \emph{canceling
  label history} $H^{cl}_{\alpha}$ (fifo label history of size
$\gconst{M} = (\gconst{K}+1)\gconst{K}^{cl}$), the \emph{accepted
  proposal record} $r_{\alpha}$ and the \emph{label history}
$H_{\alpha}$ described as follows.  The accepted proposal record
$r_{\alpha}$ is a vector indexed by the processor identifiers.  For
each identifier $\mu$, the field $r_{\alpha}[\mu]$ contains either the
null value $\bot$ or a couple composed of a tag and consensus value.
The variable $H_{\alpha}$ is a vector indexed by the processor
identifiers.  For each identifier $\mu$, the field $H_{\alpha}[\mu]$
is a fifo label history of size $\gconst{K}$. Note that all the label
histories, and canceling label history are bounded by recent activity,
since they accumulate only a polynomial number of the latest labels.

\paragraph{\bf Tag Increment Functions.}
We define the step increment function, $\nu^s$, and the trial
increment function, $\nu^t$.  Both functions arguments are a processor
identifier $\lambda$, a tag $x$, and the canceling label history
$H_{\lambda}^{cl}$, and they both return a tag (the incremented tag).
First, a copy $y$ of the tag $x$ is created. The step increment
function then increments the step in the first valid entry of $y$ and
resets the corresponding trial field to zero.  The trial increment
function only increments the trial field in the first valid entry of
$y$.  Then, in both functions, it is checked whether the entry
$y[\lambda]$ is valid or not. If it is not, the label value
$x[\lambda].l$ is stored in the canceling label history, a new label
is produced\footnote{With the label increment function from the finite
  labeling scheme (cf. Definition~\ref{def:finite_labeling_scheme}).}
in $x[\lambda]$ with the labels in the canceling label history, and
the corresponding step and trial fields are reset to zero.

\paragraph{\bf Protocol.}
Each processor can play two roles, namely, the acceptor role and the
proposer role. A processor $\alpha$ plays both\footnote{One can think
  of having two threads on the same processor.} the acceptor role and
the proposer role as long as $\Theta_{\alpha}$ is equal to
$\gconst{true}$.  When $\Theta_{\alpha}$ is equal to $\gconst{false}$,
the processor $\alpha$ only plays the acceptor role.
The current step and trial of a processor are determined by the step
and trial values in the first valid entry of its tag.  A proposer
tries to impose some proposal for its current step.  To do so, it
executes the following two phases (cf.~Algorithm~\ref{alg:proposer}).

{\bf (Phase 1).} The proposer, say $\lambda$, reads a new proposal and
tries to recruit a quorum of acceptors by broadcasting a message
(phase 1, message $p1a$) with its tag $a_{\lambda}$
(Algorithm~\ref{alg:proposer}, line~\ref{algline:loop_bcast_p1a}).  It
waits for the replies from a majority of acceptors.  When an acceptor
$\alpha$ receives this $p1a$ message, it either adopts the proposer
tag if the proposer tag is greater than its own tag $a_{\alpha}$, or leaves
its tag unchanged otherwise.  The acceptor replies (phase 1,
message $p1b$) to the proposer with its tag (updated or not)
and the proposal, either null or a couple (tag, consensus value),
stored in its accepted proposal variable $r_{\alpha}[\chi(a_{\alpha})]$.

Upon receiving the acceptor replies, the proposer $\lambda$ knows if
it has managed to recruit a majority of acceptors. In that case, the
proposer $\lambda$ can move to the second phase.  Otherwise,
$\lambda$ has received at least one acceptor reply whose tag is not
less than or equal to the proposer tag of $\lambda$.  At each
reception of such an acceptor tag, the proposer $\lambda$ modifies
its tag in order for the proposer tag to be greater than
the acceptor tag received. When messages are received  from at
least half of the processors, the proposer begins a new
phase $1$ with its updated tag.

{\bf (Phase 2).} When the proposer $\lambda$ reaches this point, it
has managed to recruit a quorum of acceptors and it knows all the
latest proposals that they accepted for the entry $\chi(a_{\lambda})$.
Assume for instance that the proposer tag points to step $s$, i.e.,
the step value in the first valid entry $\mu$ of the proposer tag is
equal to $s$.  Then (Algorithm~\ref{alg:proposer},
line~\ref{algline:loop_p2_begin} to line~\ref{algline:loop_p2_process_replies_end}) the
proposer $\lambda$ first checks that the tags associated with the
received proposals all share the same first valid entry and the same
corresponding label as the tag of $\lambda$.  If it is not the case,
then $\lambda$ keeps its original proposal.  Otherwise, it looks for
non-null proposals for step $s$ and if there are some, it copies the
proposal with the maximum tag (among those that point to step $s$) in
its proposal variable.  If there are more than two different proposals
associated with this maximum tag, then $\lambda$ keeps its original proposal.

Next, the proposer $\lambda$ sends to all the acceptors a message
(phase 2, message $p2a$) containing its tag along with the
proposal it has computed (Algorithm~\ref{alg:proposer},
line~\ref{algline:loop_bcast_p2a}) and waits for the replies of a
majority of acceptors.  When an acceptor $\alpha$ receives this $p2a$
message, if the proposer tag is greater than or equal to its own
tag, then the acceptor adopts the proposer tag and stores the proposal
in the variable $r_{\alpha}$.  Otherwise, the acceptor leaves its
 tag and the accepted proposals record unchanged.  Next, it
replies (phase $2$, message $p2b$) to the proposer with its 
tag (updated or not).

After having received the replies from a majority of acceptors, the
proposer $\lambda$ knows if a majority have accepted its proposal.  In
that case, it broadcasts a decision message containing its proposer
tag and the successful consensus value (Algorithm~\ref{alg:proposer},
line~\ref{algline:loop_bcast_dec}). At the reception of this message,
any acceptor with a tag less than or equal to the proposer
tag decides on the given proposal. The proposer $\lambda$ can then
move to the next step.  Otherwise, the proposer $\lambda$ has received
tags that are not less than or equal to the proposer tag, and thus
$\lambda$ updates its proposer tag accordingly, and starts a new
phase $1$.

\paragraph{\bf Precisions.}
By ``$\alpha$ adopts the tag $b$'', we mean that $\alpha$ copies the
content of the first valid entry in $b$ to the same entry in
$\alpha$'s acceptor tag\footnote{Note that only the entry
  $a_{\alpha}[\chi(b)]$ is modified. In fact, we have $a_{\alpha}
  \simeq b$ and not $a_{\alpha} = b$.}, i.e., $a_{\alpha}[\chi(b)]
\leftarrow b[\chi(b)]$.  Furthermore, every time a processor
$\alpha$ modifies its tag, it also does the following.  If the label
$l$, in some entry $\mu$ of the tag, is replaced by a new label,
then the label $l$ is stored in the label history $H_{\alpha}[\mu]$
that corresponds to the identifier $\mu$ and a label that cancels
the new label is looked for in the (bounded) label history $H_{\alpha}[\mu]$,
updating the corresponding canceling field accordingly. If the label
in the entry $\alpha$, i.e., the only entry in which the proposer
$\alpha$ can create a label, gets canceled, then the associated
canceling label is stored in the (bounded) canceling label history
$H_{\alpha}^{cl}$.  Any new label produced in the entry $\alpha$ of
the tag at processor $\alpha$ is also stored in $H_{\alpha}^{cl}$.
In addition, for every $\mu$, the accepted proposal
$r_{\alpha}[\mu]$ is cleared, i.e., $r_{\alpha}[\mu] \leftarrow
\bot$, whenever there is a label change in the entry
$a_{\alpha}[\mu]$. A non-null field $r_{\alpha}[\mu] = (b,p)$ is
also cleared whenever the label in the entry $b[\mu]$ is different
than the label in the entry $a_{\alpha}[\mu]$, or the labels are
equal but the entry $b[\mu]$ is lexicographically greater than the
entry $a_{\alpha}[\mu]$.  Finally, any processor $\alpha$ always
checks that the entry $\alpha$ of its tag is valid.  If it is not,
the corresponding label is stored in the (bounded) canceling label history, a
new label is produced instead and the step and trial fields are
reset to zero.

%%%%%%%%%%%%%%%%%%%%%%%%%%%%%%%%%%%%%%%%%%%%%%%%%%%%%%%%%%%%%%%%%%%
\vspace*{-.2cm}
\section{Proof in a Nutshell}
\label{sec:results}
\vspace*{-.2cm}

In this section, we present a summary of the main results of this
work. Full details on the definitions, theorems and proofs are given
in Appendix~\ref{app:proofs}.
An \emph{epoch} at processor $\lambda$ is a maximal local subexecution
during which the first valid entry of its tag and the corresponding
label remains constant
(Appendix~\ref{app:tag_stabilization_definitions},
Definition~\ref{def:epoch}).  Given a bounded integer $h$, an
\emph{$h$-safe epoch} at processor $\lambda$ is an epoch at $\lambda$
that ends because the step or trial values in the first valid entry
$\mu$ of its tag have reached the maximum value $2^{\mathfrak{b}}$. In
addition, in the configuration of the system that precedes this epoch,
for any tag in the system, either the label in the entry $\mu$ is
different than the one used by $\lambda$, or the corresponding step
and trial values are less than $h$.  From the point of view of
$\lambda$, within an $h$-safe epoch, everything seems like an original
Paxos execution initialized with integer values less than $h$.  For
instance, we understand that for $\lambda$ to jump, e.g., from step
$10$ to $15$ within a $0$-safe epoch, there must be a chain of events
totally ordered by the happen-before relation that correspond to
decision for steps $10$ to $15$.  Thus, a $h$-safe epoch is actually
as long as counting from $h$ to $2^{\mathfrak{b}}$. The first main
result (Appendix~\ref{app:tag_stabilization_results},
Theorem~\ref{th:existence_safe_epoch}) states that there is some
proposer $\lambda$ at which there exists a $0$-safe epoch.  Note that
this safe epoch is not necessarily unique, and it is not necessary to
wait for it.  Indeed, this results simply states that one has not to
worry about having only very short epochs at processor $\lambda$.

The second part of this work highlights the link between such a safe
epoch at $\lambda$ and the safety property on the global system. The
idea is that $\lambda$ is talking to quorums whose members cannot
alter the first valid entry nor the corresponding label of $\lambda$
during $\sigma$, and must use the same first valid entry and
corresponding label. Roughly saying, within a globally defined set of
events related to the $h$-safe epoch at $\lambda$, we show that for
any two decision events for the same step $s \geq h$ the two decided
proposals are equal (Appendix~\ref{app:practical_safety_results},
Theorem~\ref{th:practical_safety}).
% The set of unsafe step values is
% related to the existence of fake messages in the communication
% channels and is proved to be bounded above by the total capacity of
% the communication channels. If the upper bound $2^{\mathfrak{b}}$ is
% large enough relatively to this bound, then there is at least one
% very long sequence of safe steps.

These two results rely on a proper management of the labels.  Indeed,
the comparison relation on labels is not transitive; there might be
cycles of labels. The algorithm uses histories of labels to detect
such cycles. Precisely, the entry $\lambda$ of the tag of a processor
$\alpha$ is associated with the label history
$H_{\alpha}[\lambda]$. Whenever, the corresponding label is replaced
by a new label, the old label is stored in the label history, and a
canceling label for the new label is looked for. This technique
prevents the label field to follow a cycle whose length is less than
the size of the label history. However, the size of the label history
is chosen to equal the total label capacity of the system which
implies that longer cycles are possible in the entry $\lambda$ if and
only if the processor $\lambda$ produces at least one label meanwhile
(Appendix~\ref{app:tag_stabilization_results},
Lemma~\ref{lem:cycle_labels}).  Thanks to this technique, it is
possible to order events relatively to epochs occurring at $\lambda$
since labels are produced by $\lambda$ only at the end of some epochs
occurring at $\lambda$ (Appendix~\ref{app:practical_safety_results},
Lemma~\ref{lem:epoch_cycle_labels}).  If an epoch is practically
infinite, it gives a way to discard events that happen after this
epoch.

Besides, to guarantee liveness, the original Paxos algorithm requires
a single proposer for each step; knowing that two steps may have
different attributed proposers. In our model, an external module
called $\Theta$ is responsible for selecting the processors that act
as proposers. For the tag system to stabilize, it is only needed that
at least one processor acts infinitely often as a proposer.
Nevertheless, this external module is generally implemented with a
failure detector.  For the sake of completeness we present a simple
implementation of a self-stabilizing failue detector.% (Section~\ref{sec:liveness}).

\vspace*{-.2cm}
\section{Self-Stabilizing Failure Detector}
\label{sec:liveness}
\vspace*{-.2cm}

Liveness for some step $s$ in Paxos is not guaranteed unless there is
a unique proposer for this step $s$.  The original Paxos algorithm
assumes that the choice of a distinguished proposer for a given step
is done through an external module.  In the sequel, we present an
implementation of a self-stabilizing failure detector that works under
a partial synchronism assumption. Note that this assumption is strong
enough to implement a perfect failure detector, but a perfect failure
detector is not mandatory for our algorithm to converge (i.e., the tag
system). This brief section simply explain how a
\emph{self-stabilizing} implementation can be done; which is, although
not difficult, not obvious either.  Each processor $\alpha$ has a
vector $L_{\alpha}$ indexed by the processor identifiers; each entry
$L_{\alpha}[\mu]$ is an integer whose value is comprised between $0$
and some predefined maximum constant $W$.  Every processor $\alpha$
keeps broadcasting a hearbeat message $\langle hb, \alpha \rangle$
containing its identifier (e.g., by using
\cite{dolev2000self,DBLP:conf/sss/DolevHSS12}). When the processor
$\alpha$ receives a heartbeat from processor $\beta$, it sets the
entry $L_{\alpha}[\beta]$ to zero, and increments the value of every
entry $L_{\alpha}[\rho]$, $\rho \neq \beta$ that has value less than
$W$.  The detector output at processor $\alpha$ is the list
$F_{\alpha}$ of every identifier $\mu$ such that $L_{\alpha}[\mu] =
W$. In other words, the processor $\alpha$ assesses that the processor
$\beta$ has crashed if and only if $L_{\alpha}[\beta] = W$.

\emph{(Interleaving of Heartbeats). For any two live processors
  $\alpha$ and $\beta$, between two receptions of heartbeat $\langle
  hb, \beta \rangle$ at processor $\alpha$, there are strictly less
  than $W$ receptions of heartbeats from other processors.}  Under
this condition, for every processor $\alpha$, if the processor $\beta$
is alive, then eventually the identifier $\beta$ does not belong to
the list $F_{\alpha}$.  The connection with the external module
$\Theta$ in Section~\ref{sec:model} can be defined as follows: $
\Theta_{\alpha} = \gconst{true} \Leftrightarrow \alpha = \min(\mu;~
L_{\alpha}[\mu] < W )$.  Under this hypothesis,
we see that the module $\Theta$ eventually satisfies the conditions in
Assumption~\ref{hyp:module_theta}, Section~\ref{sec:model}.

% \begin{equation}
%   \Theta_{\alpha} = \gconst{true} 
%   \Leftrightarrow \alpha = \min(\mu;~ L_{\alpha}[\mu] < W )
% \end{equation}  

%%%%%%%%%%%%%%%%%%%%%%%%%%%%%%%%%%%%%%%%%%%%%%%%%%%%%%%%%%%%%%%%%%%
\vspace*{-.3cm}
\section{Conclusion}
\label{sec:conclusion}
\vspace*{-.2cm}

The original Paxos algorithm provides a solution to the problem, for a
distributed system, to reach successively several consensus on
different requests to apply.  A proper tagging system using natural
integers is defined so that, although the liveness property, i.e., the
fact that, in every consensus instance, every processor eventually
decides, is not guaranteed, the safety property is ensured: no two
processors decide on different values in the same consensus instance.
The original formulation, however, does assume a consistent initial
state and assumes that consistency is preserved forever by applying
step transitions from a restricted predefined set of step transitions.
This line of consistency preserving argument is fragile and error
prone in any concrete system that should exhibit availability and
functionality during very long executions. Hence, there is an urgent
need for self-stabilizing on-going systems, and in particular for the
very heart of asynchronous replicated state machine systems used by
the leading companies to ensure robust services.  One particular
aspect of self-stabilizing systems is the need to re-examine the
assumption concerning the use of (practically) unbounded time-stamps.
While in practice it is reasonable for Paxos to assume that a bounded
value, represented by $64$ bits, is a natural (unbounded) number, for
all practical considerations, in the scope of self-stabilization the
$64$ bits value may be corrupted by a transient fault to its maximal
value at once, and still recovery following such a transient fault
must be guaranteed. More generally, the designer of self-stabilizing
systems, does not try to protect its system against specific ``bad''
scenarios. She assumes that some transient faults, whatever their
origin is, corrupt (a part of) the system and ensures that the system recovers
automatically after such fault occurrences.

Using a finite labeling scheme, we have defined a new kind of tag
system that copes with such transient faults.  The tag is defined as a
vector indexed by the processor identifiers, such that each entry
contains a label, a step and a trial value.  Incrementing the label
becomes a way to properly reset the step and trial values in a given
entry of a tag. Each processor is responsible for producing labels
only in the entry that corresponds to its identifier.  Therefore, once
it collects enough information about the labels present in its
attributed entry, a processor is able to produce a label that no other
processor can cancel. Hence, in a tag, there might be several entries
with ``winning'' labels, and the owner of the tag uses the entry with
the lowest identifier.  Our algorithm ensures that at some point in
time, almost all the processors uses the same entry, the same
corresponding label and integer (step and trial) fields with low
values. From this point, the system behaves like the original Paxos
until the maximum value $2^{\mathfrak{b}}$ is reached by some step or
trial variables. This is what we named a ``practically
self-stabilizing'' behaviour, since the length of the stable execution
is not infinite as in classical self-stabilization but long enough for
any concrete system's timescale\footnote{Recall that counting from 0
  to $2^{64}$ by incrementing every nanoseconds lasts about 500
  years.}, just as assumed in the original Paxos algorithm.

\newpage
\bibliographystyle{abbrv}
\bibliography{biblio}

%%%%%%%%%%%%%%%%%%%%%%%%%%%%%%%%%%%%%%%%%%%%%%%%%%%%%%%%%%%%%%%%%%%
\newpage
\appendix
\section{Tag System - Formal Definitions}
\label{app:tag_system_definitions}

\begin{definition}[Bounded Integer]
  \label{def:bounded_integer}
  Given a positive integer $\mathfrak{b}$, a $\mathfrak{b}$-bounded
  integer, or simply a bounded integer, is any non-negative integer
  less than or equal to $2^{\mathfrak{b}}$.
\end{definition}

\begin{definition}[Finite Labeling Scheme]
  \label{def:finite_labeling_scheme}
  A \emph{finite labeling scheme} is a $4$-tuple $\overline{\mathcal{L}} =
  (\mathcal{L}, \prec, d, \nu)$ where $\mathcal{L}$ is a finite set
  whose elements are called \emph{labels}, $\prec$ is a partial
  relation on $\mathcal{L}$ that is irreflexive ($l \not\prec l$) and
  antisymmetric ($\not\exists (l,l') ~ l \prec l' \land l' \prec l$),
  $d$ is an integer, namely the \emph{dimension} of the labeling
  scheme, and $\nu$ is the \emph{label increment function}, i.e., a
  function that maps any finite set of at most $d$ labels to a label
  such that for every subset $A$ in $\mathcal{L}$ of at most $d$ labels, 
  for every label $l$ in $A$, we have $l \prec \nu(A)$.
  We denote the reflexive closure of $\prec$ by $\preccurlyeq$.  
\end{definition}

\begin{remark}
  The definition of a finite labeling scheme imposes that the relation
  $\prec$ is not transitive. Hence, it is not an order relation.
\end{remark}

\begin{definition}[Canceling Label]
  Given a label $l$, a \emph{canceling label} for $l$ is a label $cl$
  such that $cl \not\preccurlyeq l$.
\end{definition}

\begin{definition}[Tag System]
  \label{def:tag_system}
  A tag system is given by a $4$-tuple $(\mathfrak{b}, \Pi, \omega,
  \overline{\mathcal{L}})$ where $\mathfrak{b}$ is positive integer,
  $\Pi$ is the totally ordered finite set of processor identifiers,
  $\omega$ is a special symbol such that $\omega \not\in \Pi$ and
  $\overline{\mathcal{L}}$ is a finite labeling scheme. In addition
  the order on $\Pi$ is extended as follows: for every $\mu \in \Pi$,
  $\mu < \omega$.
\end{definition}

\begin{definition}[Tag]
  Given a tag system $(\mathfrak{b}, \Pi, \omega,
  \overline{\mathcal{L}})$, a \emph{tag} is a vector $a[\mu] =
  (l~s~t~id~cl)$ where $\mu$ and $id$ are processor identifiers, $l$ is a label,
  $cl$ is either the null value noted $\bot$ or a canceling label for
  $l$, and $s$ and $t$ are $\mathfrak{b}$-bounded integers
  respectively called the \emph{step} and \emph{trial} fields.  The
  \emph{entry indexed by} $\mu$ in the tag $a$, or simply the entry
  $\mu$ in $a$, refers to the entry $a[\mu]$.  The entry $\mu$ is said
  to be \emph{valid} when the corresponding canceling field is null,
  $a[\mu].cl = \bot$, and both the corresponding step and trial values
  are strictly less than the maximum value, i.e., $a[\mu].s <
  2^{\mathfrak{b}}$ and $a[\mu].t < 2^{\mathfrak{b}}$.
\end{definition}

\begin{definition}[First Valid Entry]
  \label{def:first_valid_entry}
  Given a tag $a$, the first valid entry in the tag is defined by
  \begin{equation*}
    \chi(a) = \min \left( \{ \mu \in \Pi ~|~ a[\mu] \textrm{ is valid}\} \cup \{\omega\} \right)
  \end{equation*}
\end{definition}

\begin{definition}[Comparison of Tags]
  \label{def:tag_comparison}
  Given two tags $a$ and $a'$, we note $a \prec a'$ when either
  $\chi(a) > \chi(a')$ or $\chi(a) = \chi(a') = \mu < \omega$
  and\footnote{Lexicographical comparison using
    the corresponding relation on labels, integers and processor
    identifiers.} $a[\mu].(l~s~t~id) < a'[\mu].(l~s~t~id)$.
  We note $a \simeq a'$ when $\chi(a) = \chi(a')$ and $a[\chi(a)] = a'[\chi(a)]$.
  We note $a \preccurlyeq a'$ when either $a \prec a'$ or $a \simeq a'$.
\end{definition}

\section{Construction of a Finite Labeling Scheme}
\label{app:construction_labeling_scheme}

We show how to construct a finite labeling scheme $(\mathcal{L}, \prec, d, \nu)$.
First, consider the set of integers $X = \{1,2,...,K\}$ with $K = d^2+1$.
We define the set $\mathcal{L}$ to be the set of every tuple $(z, A)$ where $z \in X$ is the \emph{sting}, and
$A \subset X$ with $|A| \leq d$ is called the \emph{antistings}.
The relation $\prec$ is defined as follows
\begin{equation}
  l = (z,A) \prec l' = (z',A') \Leftrightarrow (z \in A') \land (z' \not\in A)
\end{equation}
The function $\nu$ is defined as follows.
Given $r$ labels $(s_1,A_1)$, \dots, $(s_r,A_r)$ with $r \leq d$, the label $\nu(l_1,\dots,l_r) = (s,A)$
is given by
\begin{align}
  s &= \min\left\{ X - (A_1 \cup \dots \cup A_r) \right\} \\
  A &= \{s_1,\dots,s_r\}
\end{align}
The function is well-defined since $r \leq d$ and $| A_1 \cup \dots \cup A_r | \leq d^2 < |X|$.
In addition, for every $i$, we have $s \not\in A_i$ and $s_i \in A$, thus $(s_i,A_i) \prec (s,A)$.

\section{Algorithm Details}
\label{app:algorithm_details}

\SetEndCharOfAlgoLine{}

We give more details about the algorithms.
We consider a tag system $(\mathfrak{b}, \Pi, \omega, \mathcal{L},
\prec, d, \nu)$ such that $\Pi$ is the set of processor
identifiers and the labeling scheme dimension is equal to $(\gconst{K}
+ 1)\gconst{K}^{cl}$.

\subsection{Tag Procedures}
\label{subsec:tag_procedures}

Algorithm~\ref{alg:tagprocs2} defines a procedure \texttt{clean} that
cleans the canceling fields of a given tag as follows.  The procedure
takes as input a processor identifier $\lambda$ and a tag $a$.  After
the completion of the procedure, for every entry $\mu$ in the tag $a$,
if the canceling field $a[\mu].cl$ is not null, then its value is
a canceling label for the label in $a[\mu].l$. In addition, every
identifier value in $a[\mu].id$ is equal to $\lambda$.  The second
procedure \texttt{fill\_cl} updates the canceling fields of two given
tags $x$ and $y$ as follows. After the completion of the
procedure, for any $\mu \in \Pi$, if the label $x[\mu].l$ or
$x[\mu].cl$ (not equal to $\bot$) cancels $y[\mu].l$, then
$y[\mu].cl$ is not null.  And if $x[\mu].l = y[\mu].l$ with one of the
integer fields in $x[\mu]$ being equal to the maximum value
$2^{\mathfrak{b}}$, then both step and trial fields $y[\mu].cl$ are
equal to $2^{\mathfrak{b}}$.  The previous remarks also hold when
exchanging $x$ and $y$.

Algorithm~\ref{alg:tagprocs3} defines the function
\texttt{check\_entry} whose arguments are a processor identifier
$\lambda$, a tag $x$, and an history of labels $L$. This function
checks whether the entry $x[\lambda]$ is valid or not.  If this entry
is invalid, it stores the label value $x[\lambda].l$ in the history
$L$, produces\footnote{With the label increment function from the
  finite labeling scheme
  (cf. Definition~\ref{def:finite_labeling_scheme}).} a new label in
$x[\lambda]$ with the labels in the history $L$ and resets the step
and trial fields to zero.  Algorithm~\ref{alg:tagprocs3} also defines
the step increment function, $\nu^s$, and the trial increment
function, $\nu^t$.  Both functions arguments are a processor
identifier $\lambda$, a tag $x$, and a fifo history of labels $L$, and
they both return a tag (the incremented tag).  First, a copy $y$ of
the tag $x$ is created.  Then the tag $y$ is cleaned
with the procedure \texttt{clean}.  The step increment function then
increments the step in the first valid entry of $y$ and resets the
corresponding trial field to zero.  The trial increment function only
increments the trial field in the first valid entry of $y$.  Then, in
both functions, it is checked whether the entry $y[\lambda]$ is valid
or not, and updated accordingly thanks to the function
\texttt{check\_entry}. Both functions return the tag $y$.

\subsection{Protocol}

We focus on the reception of a proposer message by an acceptor
(Algorithm~\ref{alg:acceptor}).  Say an acceptor $\alpha$ receives a
message $\langle p1a, \lambda, b \rangle$ from proposer $\lambda$. The
acceptor $\alpha$ first records in the canceling label history
$H^{cl}_{\alpha}$ any label in the entry $b[\alpha]$ that cancels the
label $a_{\alpha}[\alpha].l$ in the acceptor tag
(line~\ref{algline:acceptor_p1a_update_hcl}).  Using the procedure
\texttt{fill\_cl} presented in Algorithm~\ref{alg:tagprocs2}, the
acceptor $\alpha$ updates the canceling fields of both tags
$a_{\alpha}$ and $b$. Then, it checks the validity of the entry
$a_{\alpha}[\alpha]$ with the procedure \texttt{check\_entry} and
updates it accordingly (line~\ref{algline:acceptor_p1a_fill_cl}). If
the updated tags satisfy $a_{\alpha} \prec b$, then $\alpha$ adopts
the tag $b$, i.e., it copies the content of the first valid entry
$b[\chi(b)]$ to the entry $a_{\alpha}[\chi(b)]$ in $a$
(line~\ref{algline:acceptor_p1a_adopt}). If there has been a change of
label in the entry $a_{\alpha}[\chi(b)]$, then the accepted proposal
variable $r_{\alpha}[\chi(b)]$ is cleared, the old label is stored in
the history $H_{\alpha}[\chi(b)]$, and $\alpha$ looks in this history
for labels that cancel the new label $a_{\alpha}[\chi(b)].l$, updating
the corresponding canceling field accordingly
(lines~\ref{algline:acceptor_p1a_label_change1}
to~\ref{algline:acceptor_p1a_label_change2}).  Next, the acceptor
checks for every identifier $\mu$ if either the tag $b$ in the
accepted proposal $r_{\alpha}[\mu]$ uses a label different than the
label in the entry $a_{\alpha}[\mu]$, or if the tuple
$a_{\alpha}[\mu].(l~s~t~id)$ is less than the tuple
$b[\mu].(l~s~t~id)$; in such a case, the entry $r_{\alpha}[\mu]$ is
cleared. In any case, whether it
adopts the tag $b$ or not, the acceptor $\alpha$ replies to the
proposer $\lambda$ with a message $\langle p1b, \alpha, a_{\alpha},
r_{\alpha}[\chi(a_{\alpha})] \rangle$ where $a_{\alpha}$ is its
updated (or not) acceptor tag and $r_{\alpha}[\chi(a_{\alpha})]$ is
the lastly accepted proposal for the entry $\chi(a_{\alpha})$
(line~\ref{algline:acceptor_p1a_reply}).

When an acceptor $\alpha$ receives a $p2a$ message or a decision
message containing a proposal $(b,p)$, the procedure is similar. It
first updates the canceling label history $H^{cl}_{\alpha}$, the
canceling fields of $a_{\alpha}$ and $b$, and checks the validity of
the entry $a_{\alpha}[\alpha]$
(lines~\ref{algline:acceptor_p2a_update_hcl}
and~\ref{algline:acceptor_p2a_fill_cl}). The difference with the
previous case is that the condition to accept the proposal $(b,p)$ is
$a_{\alpha} \preccurlyeq b$.  In this case, the acceptor $\alpha$
adopts the tag $b$, updating the canceling field and the label history
as in the case of a $p1a$ message, stores the couple $(b,p)$ in its
accepted proposal variable $r_{\alpha}[\chi(b)]$ and, in case of a
decision message, decides on the couple $(b,p)$
(lines~\ref{algline:acceptor_p2a_accept}
and~\ref{algline:acceptor_decide}).  In addition, if there has been a
change of label in the entry $a_{\alpha}[\chi(b)]$, then\footnote{In
  this case, the variable $r_{\alpha}[\chi(b)]$ is not cleared.} the
old label is stored in the history $H_{\alpha}[\chi(b)]$, and $\alpha$
looks in this history for labels that cancel the new label
$a_{\alpha}[\chi(b)].l$, updating the corresponding canceling field
accordingly (lines~\ref{algline:acceptor_p2a_label_change1}
and~\ref{algline:acceptor_p2a_label_change2}).  We say that the
acceptor $\alpha$ has accepted the proposal $(b,p)$.  Next, the
acceptor checks for every identifier $\mu$ if either the tag $b$ in
the accepted proposal $r_{\alpha}[\mu]$ uses a label different than
the label in the entry $a_{\alpha}[\mu]$, or if the tuple
$a_{\alpha}[\mu].(l~s~t~id)$ is less than the tuple
$b[\mu].(l~s~t~id)$; in such a case, the entry $r_{\alpha}[\mu]$ is
cleared. In case of a
$p2a$ message, whether it accepts the proposal or not, the acceptor
$\alpha$ replies to the proposer $\lambda$ with a message $\langle
p2b, \alpha, a_{\alpha} \rangle$ containing its updated (or not)
acceptor tag (line~\ref{algline:acceptor_p2a_reply}).  In case of a
decision message, the acceptor does not reply.

At the end of any phase, a proposer executes a procedure named the
preempting routine (Algorithm~\ref{alg:preemptingroutine}) that mainly
consists in waiting for the replies from a majority of acceptors and
suitably incrementing the proposer tag.  The phase is considered
successful if the routine returns $ok$ and failed otherwise.  In this
routine, the processor $\lambda$ waits for $\gconst{n}-\gconst{f}$
replies from the acceptors.  Note that, although the pseudo-code
suggests $\lambda$ receives only acceptor replies
(Algorithm~\ref{alg:preemptingroutine},
line~\ref{algline:pr_receive}), the processor $\lambda$, as an
acceptor, also processes messages ($p1a$ or $p2a$) from other
proposers. The variable $a_{sent}$ stores the value of $a_{\lambda}$
that $\lambda$ has sent at the beginning of the phase.  The variable
$b$ is an auxiliary variable that helps filtering messages and is
reset to $a_{sent}$ at the beginning of each new loop
(line~\ref{algline:pr_begin_loop}).  For each message with tag
$a_{\alpha}$ and proposal $r_{\alpha}$ received from a processor
$\alpha$, the procedure updates the canceling fields of both $b$ and
$a_{\alpha}$ (line~\ref{algline:pr_alpha_b_fill_cl}).

If the current phase is a phase $1$, then a reply is considered
positive when the acceptor $\alpha$ has adopted the tag $\lambda$
sent, i.e., when\footnote{Recall that $a_{\alpha} \simeq b$ means
  $\chi(a_{\alpha}) = \chi(b)$ and $a_{\alpha}[\chi(b)] = b[\chi(b)]$.}
$a_{\alpha} \simeq b$. If the current phase is a phase $2$, the reply
is considered positive when the acceptor $\alpha$ has adopted the tag
$\lambda$ has sent and has accepted the corresponding proposal, i.e,
$a_{\alpha} \simeq b$ and $p_{\lambda} = r_{\alpha}[\chi(b)].p$.  The condition
$C^+$ (line~\ref{algline:pr_cplus}) summarizes these two cases.  A
reply is considered negative when the received acceptor tag is not
less than or equal to the tag the proposer $\lambda$ has sent, i.e., an
 acceptor tag $a_{\alpha}$ such that $a_{\alpha}
\not\preccurlyeq b$ (condition $C^-$, line~\ref{algline:pr_cminus}).
The procedure discards any acceptor reply that does not satisfy the
conditions $C^+$ nor $C^-$.  The variable $M$ counts the number of
positive replies.  The routine returns $ok$ if all the replies are
positive, i.e., $M = \gconst{n} - \gconst{f}$, and $nok$ otherwise
(lines~\ref{algline:pr_return_nok} and~\ref{algline:pr_return_ok}).

At each negative reply received, the routine updates the variable
$a_{\lambda}$ so that it is always greater than the tag received.
Precisely, it updates the canceling label history $H_{\lambda}^{cl}$
(line~\ref{algline:pr_update_hcl}), the canceling fields of
$a_{\alpha}$ and $a_{\lambda}$
(line~\ref{algline:pr_alpha_lambda_fill_cl}) and checks the validity
of the entry $a_{\lambda}[\lambda]$
(line~\ref{algline:pr_alpha_lambda_check_entry}). Recall that this implies $\chi(a_{\lambda}) \leq \lambda$. Then, the routine
checks if $a_{\alpha}$ is less than or equal to $a_{\lambda}$.  If it
is so, then the routine does not modify $a_{\lambda}$.  Otherwise
(lines~\ref{algline:pr_negreply_1} to~\ref{algline:pr_negreply_2}), it
checks if $a_{\alpha}$ has its first valid entry located at the left
of $a_{\lambda}$'s first valid entry, i.e., $\chi(a_{\alpha}) <
\chi(a_{\lambda})$.  In that case, the content of the entry
$a_{\alpha}[\chi(a_{\alpha})]$ is copied\footnote{Note that since the
  canceling fields have been updated with the procedure
  \texttt{fill\_cl}, necessarily the labels
  $a_{\alpha}[\chi(\alpha)].l$ and $a_{\lambda}[\chi(\alpha)].l$ are
  different.} to the entry $a_{\lambda}[\chi(a_{\alpha})]$ and the
trial value is incremented. In addition, the previous label in
$a_{\lambda}[\chi(\alpha)].l$ is stored in the label history
$H_{\lambda}[\chi(\alpha)]$ and possible canceling labels for the new
label in $a_{\lambda}[\chi(\alpha)].l$ are searched for in
$H_{\lambda}[\chi(\alpha)]$ (lines~\ref{algline:pr_mu_less_lambda1}
to~\ref{algline:pr_mu_less_lambda2}).  If the first valid entry
$\chi(a_{\alpha})$ in $a_{\alpha}$ is not located at the left of
$a_{\lambda}$'s first valid entry, then necessarily $\chi(a_{\alpha})
= \chi(a_{\lambda}) = \mu$, since $a_{\alpha} \not\preccurlyeq
a_{\lambda}$.  In that case, the routine compares the content of the
entries indexed by $\mu$ in $a_{\alpha}$ and $a_{\lambda}$
(lines~\ref{algline:pr_mu_eq_chi} to~\ref{algline:pr_negreply_2}).
Note that, since the routine has updated the canceling fields, the
corresponding labels are equal\footnote{Otherwise, one would cancel
  the other and contradict the definition of the first valid
  counter.}.  If both entries $a_{\alpha}[\mu]$ and $a_{\lambda}[\mu]$
share the same step value, then $a_{\lambda}$ is updated with the time
increment function $\nu^t$ (line~\ref{algline:pr_trial_inc}).
Otherwise, the step increment function is used
(line~\ref{algline:pr_negreply_2}).

% \section{Pigeonhole Principle}
% \label{app:pigeonhole_principle}

% \begin{lemma}
%   Consider a sequence $u = (u^i)_{1 \leq i \leq N}$ such that $\forall 1 \leq i \leq N, u^i \in \{0,1\}$, and $N = (n+1)m$
%   for some $n, m \in \mathbb{N}-\{0\}$.
%   Assume that the cardinal of $\{ i ~|~ u^i = 1\}$ is less than or equal to $n$.
%   Then there exists $1 \leq i_0 \leq N$ such that for every $i_0 \leq i \leq i_0 + m-1$, $u^i = 0$.
% \end{lemma}

% \begin{proof}
%   Divide the sequence $u$ in successive subsequences $\sigma^j$, $1 \leq j \leq n+1$ such that each $\sigma^j$
%   length is $m$. If for every $1 \leq j \leq n+1$, the sequence $\sigma^j$ contains at least one $1$, 
%   then the number of $1$ appearing in $u$ is at least $n + 1$, which leads to a contradiction.
%   Hence, there is some $j_0$ such that the sequence $\sigma^j$ only contains $0$.\gqed
% \end{proof}

\section{Proofs}
\label{app:proofs}

\subsection{Basics}
\label{app:basics}

\begin{lemma}[Pigeon-hole Principle]
  Consider a sequence $u = (u^i)_{1 \leq i \leq N}$ such that $\forall 1 \leq i \leq N, u^i \in \{0,1\}$, and $N = (n+1)m$
  for some $n, m \in \mathbb{N}-\{0\}$.
  Assume that the cardinal of $\{ i ~|~ u^i = 1\}$ is less than or equal to $n$.
  Then there exists $1 \leq i_0 \leq N$ such that for every $i_0 \leq i \leq i_0 + m-1$, $u^i = 0$.
\end{lemma}

\begin{proof}
  Divide the sequence $u$ in successive subsequences $\sigma^j$, $1 \leq j \leq n+1$ such that each $\sigma^j$
  length is $m$. If for every $1 \leq j \leq n+1$, the sequence $\sigma^j$ contains at least one $1$, 
  then the number of $1$ appearing in $u$ is at least $n + 1$, which leads to a contradiction.
  Hence, there is some $j_0$ such that the sequence $\sigma^j$ only contains $0$.\gqed
\end{proof}

\begin{lemma}
  \label{lem:finitephase}
  Any phase of the proposer algorithm eventually  ends.
\end{lemma}

\begin{proof}
  Let $\phi$ be a phase executed by some proposer $\lambda$.  At the
  beginning of $\phi$, the proposer $\lambda$ has broadcast a message
  with its proposer tag $a_{\lambda}$, along with a consensus value
  $p$ in case of a $p2a$ message.  Assumption~\ref{hyp:resilience}
  (Section~\ref{sec:model}) ensures that at least $\gconst{n}-\gconst{f}$ acceptors eventually
  reply.  The only reason why $\phi$ would be endless is $\lambda$
  discarding real replies from these acceptors in the preempting
  routine.  For each such acceptor $\alpha$, when it receives the
  message sent by $\lambda$, it first updates the canceling fields in
  $a_{\alpha}$ and $a_{\lambda}$.  Let $a,b$ respectively be the
  updated versions of $a_{\alpha}, a_{\lambda}$, and $\mu = \chi(a)$.
  According to the acceptor Algorithm~\ref{alg:acceptor}, if the
  acceptor $\alpha$ adopts the tag $b$ then we have $a \simeq b$, and
  in case of a $p2a$ message, it also accepts the consensus value,
  i.e., $r_{\alpha}[\chi(a)] = (b,p)$; otherwise, we must have $a \not\preccurlyeq
  b$.  These two cases correspond exactly to the conditions $C^+$ and
  $C^-$ in the Algorithm~\ref{alg:preemptingroutine}.  In other words,
  real replies are not discarded by $\lambda$, and since there are at
  least $\gconst{n}-\gconst{f}$ such replies, phase $\phi$ eventually
  ends.\gqed
\end{proof}

Given any configuration $\gamma$ of the system and any processor
idenditifer $\mu$, let $S(\gamma)$ and
$S^{cl}(\mu,\gamma)$ be two sets as follows. The set $S(\gamma)$ is the set of
every tag present either in a processor memory or in some message in a communication channel, in the configuration
$\gamma$. The set $S^{cl}(\mu,\gamma)$ denotes the collection of
labels $l$ such that either $l$ is the value of the label field
$x[\mu].l$ for some tag $x$ in $S(\gamma)$, or $l$ appears in the
label history $H_{\alpha}[\mu]$ of some processor $\alpha$, in the
configuration $\gamma$.

\begin{lemma}[Storage Limits]
  \label{lem:label_storage_limits}
  For every configuration $\gamma$ and every identifier $\mu$, we have
  $|S(\gamma)| \leq \gconst{K}$ and $|S^{cl}(\mu,\gamma)| \leq
  \gconst{K}^{cl}$.  In particular, the number of label values
  $x[\mu].l$ with $x$ in $S(\gamma)$ is less than or equal to
  $\gconst{K}$.
\end{lemma}

\begin{proof}
  Consider a configuration $\gamma$.  For each processor $\alpha$,
  there is one tag value (tag $a_{\alpha}$) in the processor state
  $\gamma(\alpha)$ of $\alpha$.  For each communication channel
  $(\alpha,\beta)$, there are at most $\gconst{C}$ different messages
  in the channel state $\gamma(\alpha,\beta)$; all these messages have
  one tag each.  Hence, the maximum number of tags present in the
  configuration $\gamma$ is $\gconst{n}$ plus $\gconst{C}$ times the
  number of communication channels. The network being complete, the
  number of communication channels is
  $\gconst{C}\frac{\gconst{n}(\gconst{n}-1)}{2}$, thus we have
  $\gconst{K} \geq | S(\gamma) |$.  For every $\alpha$, the maximum
  size of the history $H_{\alpha}[\mu]$ is $\gconst{K}$.  Hence, the
  size of $S^{cl}(\mu,\gamma)$ is bounded above by $\gconst{K}$
  (labels $x[\mu].l$ for $x$ in $S(\gamma)$) plus $\gconst{K}$ times
  the number of processors (labels from $H_{\alpha}[\mu]$ for every
  processor $\alpha$), i.e., $(\gconst{n}+1)\cdot\gconst{K} =
  \gconst{K}^{cl}$.\gqed
\end{proof}

\subsection{Tag Stabilization - Definitions}
\label{app:tag_stabilization_definitions}

\begin{definition}[Interrupt]
  \label{def:interrupt}
  Let $\lambda$ be any processor and consider a local subexecution
  $\sigma = (\gamma_k(\lambda))_{k_0 \leq k \leq k_1}$ at $\lambda$.
  We note $a_{\lambda}^k$ for the value of $\lambda$'s tag in
  $\gamma_k(\lambda)$.  We say that an interrupt has occurred at
  position $k$ in the local subsexecution $\sigma$ when one of the
  following happens
  \begin{itemize}
  \item $\mu < \lambda$, type $[\mu, \leftarrow]$ : the first valid
    entry moves to $\mu$ such that $\mu = \chi(a_{\lambda}^{k+1}) <
    \chi(a_{\lambda}^k)$, or the first valid entry does not change but
    the label does, i.e., $\mu = \chi(a_{\lambda}^{k+1}) =
    \chi(a_{\lambda}^k)$ and $a_{\lambda}^k[\mu].l \neq
    a_{\lambda}^{k+1}[\mu].l$.

  \item $\mu < \lambda$, type $[\mu, \rightarrow]$ : the first valid
    entry moves to $\mu$ such that $\mu = \chi(a_{\lambda}^{k+1}) >
    \chi(a_{\lambda}^k)$.

  \item type $[\lambda,\max]$ : the first valid entry is the same but
    there is a change of label in the entry $\lambda$ due to the step
    or trial value having reached the maximum value
    $2^{\mathfrak{b}}$; we then have $\chi(a_{\lambda}^{k+1}) =
    \chi(a_{\lambda}^k) = \lambda$ and $a_{\lambda}^k[\lambda].l \neq
    a_{\lambda}^{k+1}[\lambda].l$.

  \item $[\lambda,cl]$ : the first valid entry is the same but there
    is a change of label in the entry $\lambda$ due to the canceling
    of the corresponding label; we then have $\chi(a_{\lambda}^{k+1})
    = \chi(a_{\lambda}^k) = \lambda$ and $a_{\lambda}^k[\lambda].l
    \neq a_{\lambda}^{k+1}[\lambda].l$.
  \end{itemize}
  For each type $[\mu,*]$ ($\mu \leq \lambda$) of interrupt, we note
  $|[\mu,*]|$ the total number (possibly infinite) of interrupts of
  type $[\mu,*]$ that occur during the local subexecution $\sigma$.
\end{definition}

\begin{remark}
  \label{remark:interrupt}
  If there is an interrupt like $[\mu,\leftarrow]$, $\mu < \lambda$,
  occurs at position $k$, then necessarily there is a change of label
  in the field $a_{\lambda}[\mu].l$.  In addition, the new label $l'$
  is greater than the previous label $l$, i.e., $l \prec l'$.  Also
  note that, if $\chi(a_{\lambda}^k) = \lambda$, the proposer
  $\lambda$ never copies the content of the entry $\lambda$ of a
  received tag, say $a$, to the entry $\lambda$ of its proposer tag,
  even if $a_{\lambda}^k[\lambda].l \prec a[\lambda].l$. New labels in
  the entry $\lambda$ are only produced with the label increment
  function applied to the union of the current label and the canceling
  label history $H_{\lambda}^{cl}$.
\end{remark}

\begin{definition}[Epoch]
  \label{def:epoch}
  Let $\lambda$ be a processor.  An epoch $\sigma$ at $\lambda$ is a
  maximal (for the inclusion of local subexecutions) local
  subexecution at $\lambda$ such that no interrupts occur at any
  position in $\sigma$ except for the last position.  By the
  definition of an interrupt, every tag values within a given
  epoch $\sigma$ at $\lambda$ have the same first valid entry, say
  $\mu$, and the same corresponding label, i.e., for any two processor
  states that appear in $\sigma$, the corresponding tag values $a$ and
  $a'$ satisfies $\chi(a) = \chi(a') = \mu$ and $b[\mu].l =
  b'[\mu].l$. We note $\mu_{\sigma}$ and $l_{\sigma}$ for the first
  valid entry and associated label common to all the tag values in
  $\sigma$.
\end{definition}

\begin{definition}[$h$-Safe Epoch]
  \label{def:safe_epoch}
  Consider an execution $E$ and a processor $\lambda$.  Let $\Sigma$
  be a subexecution in $E$ such that the local subexecution $\sigma =
  \Sigma(\lambda)$ is an epoch at $\lambda$.  Let $\gamma^*$ be the
  configuration of the system right before the subexecution $\Sigma$,
  and $h$ be a bounded integer.  The epoch $\sigma$ is said to be
  \emph{$h$-safe} when the interrupt at the end of $\sigma$ is due to
  one of the integer fields in $a_{\lambda}[\mu_{\sigma}]$ having
  reached the maximum value $2^{\mathfrak{b}}$. In addition, for every
  processor $\alpha$ (resp. communication channel $(\alpha,\beta)$),
  for every tag $x$ in $\gamma^*(\alpha)$
  (resp. $\gamma^*(\alpha,\beta)$), if $x[\mu_{\sigma}].l =
  l_{\sigma}$ then the step and trial values in $x[\mu_{\sigma}].l$
  have values less than or equal to $h$.
\end{definition}

\begin{remark}
  If there is an epoch $\sigma$ at processor $\lambda$ such that
  $\mu_{\sigma} = \lambda$ and $\lambda$ has produced the label $l_{\sigma}$, then necessarily, at the beginning of
  $\sigma$, the step and trial value in $b_{\lambda}[\lambda]$ are
  equal to zero. However, other processors may already be using the
  label $l_{\sigma}$ with arbitrary corresponding step and trial
  values. The definition of a $h$-safe epoch ensures that the epoch is
  truly as long as counting from $h$ to $2^{\mathfrak{b}}$.
\end{remark}

\subsection{Tag Stabilization - Results}
\label{app:tag_stabilization_results}

\begin{lemma}
  \label{lem:chi_leq_lambda}
  Let $\lambda$ be any processor.  Then the first valid entry of its
  proposer tag is eventually always located at the left of the entry
  indexed by $\lambda$, i.e., $\chi(a_{\lambda}) \leq \lambda$.
\end{lemma}

\begin{proof}
  This comes from the fact that whenever the entry
  $a_{\lambda}[\lambda]$ is invalid, the processor $\lambda$ produces
  a new label in $a_{\lambda}[\lambda]$ and resets the step, trial and
  canceling field (cf. procedure \texttt{check\_entry},
  Algorithm~\ref{alg:tagprocs3}). Once $\chi(a_{\lambda}) \leq
  \lambda$, every consequent tag values is obtained either with the
  step or trial increment functions ($\nu^s$ or $\nu^t$), or by
  copying the content of a valid entry $\mu < \lambda$ of some tag to
  the entry $a_{\alpha}[\mu]$.  Hence the first valid entry remains
  located before the entry $\lambda$.\gqed
\end{proof}

\begin{remark}
  Thanks to this lemma, for every processor $\lambda$, it is now
  assumed, unless stated explicitly, that the entry
  $\chi(a_{\lambda})$ is always located before the entry $\lambda$,
  i.e., $\chi(a_{\lambda}) \leq \lambda$.
\end{remark}

\begin{lemma}[Cycle of Labels]
  \label{lem:cycle_labels}
  Consider a subexecution $E$, a processor $\lambda$ and an entry $\mu
  < \lambda$ in the tag variable $a_{\lambda}$. The label
  value in $a_{\lambda}[\mu].l$ can change during the subexecution $E$ and we
  note $(l^i)_{1 \leq i \leq T+1}$ for the sequence of successive
  distinct label values that are taken by the label $a_{\lambda}[\mu].l$ in the
  entry $\mu$ during the subexecution $E$. We assume that the first
  $T$ labels $l^1, \dots, l^T$ are different from each other, i.e.,
  for every $1 \leq i < j \leq T$, $l^i \neq l^j$.
  \begin{itemize}
  \item If $T > \gconst{K}$, then at least one of the label $l^i$ has
    been produced\footnote{Precisely, it has invoked the label
      increment function to update the entry $\mu$ of its tag
      $a_{\mu}$.} by the processor $\mu$ during $E$.
  \item If $T \leq \gconst{K}$ and $l^{T+1} = l^1$, then when the
    processor $\lambda$ adopts the label $l^{T+1}$ in the entry $\mu$ of
    its tag $a_{\lambda}$, the entry $\mu$ becomes invalid.
  \end{itemize}
\end{lemma}

\begin{proof}
  First note that a processor adopts a new label in the entry $\mu$ of
  one of its tag, only when the old label is less than the new label.
  Hence, we have for every $1 \leq i \leq T$, $l^i \prec l^{i+1}$ and,
  in particular, if $l^1 = l^{T+1}$,  $l^2 \not\preccurlyeq
  l^{T+1}$.
  Assume $T > \gconst{K}$.  Since in every configuration there is at
  most $\gconst{K}$ tags in the system, and $\mu$ is the only source
  of labels in the entry $\mu$, the fact that $\lambda$ has seen more
  than $\gconst{K}$ different label values in the entry $\mu$ is
  possible only if $\mu$ has produced at least one label during $E$.
  If $T \leq \gconst{K}$ and $l^1 = l^{T+1}$, i.e., there is a cycle of
  length $T$, then when $\lambda$ adopts the label $l^{T+1} = l^1$,
  the label history $H_{\lambda}[\mu]$ contains the whole sequence
  $l^1, \dots, l^T$ since its size is
  $\gconst{K}$. Hence, $\lambda$ sees the label $l^2$ that cancels the label
  $l^{T+1}$, and the entry $\mu$ becomes invalid.\gqed
\end{proof}

\begin{lemma}[Counting the Interrupts]
  \label{lem:countints}
  Consider an infinite execution $E_{\infty}$ and let $\lambda$ be a processor
  identifier such that every processor $\mu < \lambda$ produces
  labels finitely many times.  Consider an identifier $\mu < \lambda$
  and any processor $\rho \geq \lambda$.  Then, the local execution
  $E_{\infty}(\rho)$ at $\rho$ induces a sequence of interrupts such that
  \begin{equation}
    |[\mu,\leftarrow]| \leq R_{\mu} = (J_{\mu} + 1)\cdot(\gconst{K}+1) - 1
  \end{equation}
  where $J_{\mu}$ is the number of times the processor $\mu$ has
  produced a label since the beginning of the execution.
\end{lemma}

\begin{proof}
  We note $(a_{\rho}^k)_{k \in \mathbb{N}}$ the sequence of $\rho$'s
  tag values appearing in the local execution $E_{\infty}(\rho)$.
  Assume on the contrary that $|[\mu,\leftarrow]|$ is greater than
  $R_{\mu}$.  Note that after an interrupt like $[\mu,\leftarrow]$,
  the first valid entry $\chi(a_{\rho})$ is equal to $\mu$.  In
  particular, the entry $\mu$ is valid after such interrupts.  Also,
  the label value in the entry $a_{\lambda}[\mu].l$ does not change
  after an interrupt like $[\mu,\rightarrow]$. We define an increasing
  sequence of integers $(f(i))_{1 \leq i \leq R_{\mu}+1}$ such that
  the $i$-th interrupt like $[\mu,\leftarrow]$ occurs at $f(i)$ in the
  sequence $(a_{\rho}^k)_{k \in \mathbb{N}}$.  The sequence $l^i =
  a_{\rho}^{f(i)+1}[\mu].l$ is the sequence of distinct labels
  successively taken by $a_{\rho}[\mu].l$.  We have $l^i \prec
  l^{i+1}$ for every $1 \leq i \leq R_{\mu}$.
  
  Divide the sequence $(l^i)_{1 \leq i \leq R_{\mu}+1}$ in successive
  segments $u^j$, $1 \leq j \leq J_{\mu}+1$, of size
  $\gconst{K}+1$ each.  For any $j$, if all the $\gconst{K} + 1$
  labels in $u^j$ are different, then, by
  Lemma~\ref{lem:cycle_labels}, the processor $\mu$ has produced at
  least one label.  Since the processor $\mu$ produces labels at most
  $J_{\mu}$ many times, there is some sequence $u^{j_0}$ within which some
  label appears twice.  In other words, in $u^{j_0}$ there is a cycle
  of length less than or equal to $\gconst{K}$.  By
  Lemma~\ref{lem:cycle_labels}, this implies that the entry $\mu$
  becomes invalid after an interrupt like $[\mu,\leftarrow]$; this is a
  contradiction.\gqed
\end{proof}

\begin{theorem}[Existence of a $0$-Safe Epoch]
  \label{th:existence_safe_epoch}
  Consider an infinite execution $E_{\infty}$ and let $\lambda$ be a
  processor such that every processor $\mu < \lambda$ produces labels
  finitely many times.  We note $|\lambda|$ for the number of
  identifiers $\mu \leq \lambda$, $J_{\mu}$ for the number of times a
  proposer $\mu < \lambda$ produces a label and we define
  \begin{equation}
    T_{\lambda} = (\sum_{\mu < \lambda}R_{\mu} + 1) \cdot (|\lambda|+1) \cdot (\gconst{K^{cl}}+1) \cdot (\gconst{K} + 1)
  \end{equation}
  where $R_{\mu} = (J_{\mu} + 1)\cdot(\gconst{K}+1) - 1$.  Assume that
  there are more than $T_{\lambda}$ interrupts at processor $\lambda$
  during $E_{\infty}$ and consider the concatenation $E_c(\lambda)$ of
  the first $T_{\lambda}$ epochs, $E_c(\lambda) = \sigma^1 \dots
  \sigma^{T_{\lambda}}$.  Then $E_c(\lambda)$ contains a $0$-safe epoch.
\end{theorem}

\begin{proof}
  By Lemma~\ref{lem:countints}, we have $\sum_{\mu <
    \lambda}|[\mu,\leftarrow]| \leq \sum_{\mu < \lambda}R_{\mu}$ in
  the local execution $E_{\infty}(\lambda)$, a fortiori in the
  execution $E_c(\lambda)$.  By the pigeon-hole principle, there must
  be a local subexecution $E_1(\lambda) = \sigma^i\dots\sigma^{i+X-1}$
  in $E_c(\lambda)$, where $X = (|\lambda|+1) \cdot (\gconst{K}^{cl}+1)
  \cdot (\gconst{K} + 1)$, that contains only interrupts like
  $[\mu,\rightarrow]$, $[\lambda,\max]$ or $[\lambda,cl]$.  Naturally,
  the number of interrupts like $[\mu,\rightarrow]$ in $E_1(\lambda)$
  is less than or equal to $|\lambda|$.  Hence, another application of
  the pigeon-hole principle gives a local subexecution $E_2(\lambda) = \sigma^j\dots\sigma^{j+Y-1}$
  in $E_1(\lambda)$ where $Y = (\gconst{K}^{cl}+1)\cdot (\gconst{K}+1)$
  that contains only interrupts like $[\lambda,\max]$ or $[\lambda,cl]$.

  Assume first that within $E_2(\lambda)$, there is a subexecution
  $E_3(\lambda) = \sigma^k\dots\sigma^{k+Z-1}$ where $Z =
  \gconst{K}+1$ in which there are only interrupts like
  $[\lambda,\max]$. Since $\gconst{K}+1 \leq \gconst{M}$ the size of
  the canceling label history\footnote{Recall that the canceling label
    history also records the label produced in the entry $\lambda$.},
  we have $l_{\sigma^k}, \dots, l_{\sigma^{h-1}} \prec l_{\sigma^h}$, for
  every $k < h < k+Z$.  In particular, all the labels
  $l_{\sigma^k},\dots,l_{\sigma^{k+Z-1}}$ are different. Since $Z =
  \gconst{K}+1$ and since there is at most $\gconst{K}$ tags in a
  given configuration, there is necessarily some $k \leq h < k + Z$
  such that the label $l_{\sigma^h}$ does not appear\footnote{Note that
    $\lambda$ is the only processor to produce labels in entry
    $\lambda$, so during the subexecution that correspond to an epoch
    $\sigma^h$ at $\lambda$, the set of labels in the entry $\lambda$
    of every tag in the system is non-increasing.}  in the
  configuration $\gamma^*$ that corresponds to the last position in
  $\sigma^{h-1}$. Also, by construction, we have $\mu_{\sigma^h} =
  \lambda$ and $\sigma^h$ ends with an interrupt like
  $[\lambda,\max]$.  Hence, $\sigma^h$ is $0$-safe.

  Now, assume that there is no subexecution $E_3$ in $E_2$ as in the
  previous paragraph.  This means that if we look at the successive
  interrupts that occur during $E_2(\lambda)$, between any two
  successive interrupts like $[\lambda,cl]$, there is at most
  $\gconst{K}$ interrupts like $[\lambda,\max]$. Since the length of
  $E_2(\lambda)$ is $(\gconst{K^{cl}}+1) \cdot (\gconst{K} + 1)$,
  there must be at least $\gconst{K^{cl}}+1$ interrupts like
  $[\lambda,cl]$.  Let $E_4(\lambda)$ be the local subexecution that
  starts with the epoch associated with the first interrupt like $[\lambda,cl]$ and ends
  with the epoch associated with the interrupt $[\lambda,cl]$ numbered $\gconst{K}^{cl}$. Let
  $\sigma$ in $E_2(\lambda)$ be the epoch right after $E_4(\lambda)$.
  By construction, there is at most
  $\gconst{K}^{cl}\cdot(\gconst{K}+1)$ epochs in $E_4(\lambda)$
  which is the size $\gconst{M}$ of the history
  $H_{\lambda}^{cl}$.  Hence, at the beginning of $\sigma$, the
  history $H_{\lambda}^{cl}$ contains all the labels the processor
  $\lambda$ has produced during $E_4$ as well as all the
  $\gconst{K}^{cl}$ (exactly) labels it has received during $E_4$.
  Since there is at most $\gconst{K^{cl}}$ candidates label for
  canceling in the system, necessarily, in the first configuration of
  $\sigma$, the history $H_{\lambda}^{cl}$ contains every candidates
  label for canceling present in the whole system. And since
  $l_{\sigma}$ is greater, by construction, than every label in the
  history $H_{\lambda}^{cl}$, $l_{\sigma}$ was not present in the entry
  $\lambda$ of some tag in the configuration that precedes $\sigma$
  and it cannot be canceled by any other label present in the
  the system. In addition, by construction, $E_2$
  only contains interrupts like $[\lambda,\max]$ or
  $[\lambda,cl]$. From what we said about $l_{\sigma}$, the interrupt
  at the end of $\sigma$ is necessarily $[\lambda,\max]$. In other
  words, the epoch $\sigma$ is a $0$-safe epoch.\gqed
\end{proof}

\begin{remark}
  Note that the epoch found in the proof is not necessarily the unique
  $0$-safe epoch in $E_c(\lambda)$.  The idea is only to prove that
  there exists a practically infinite epoch.  If the first epoch
  $\sigma$ at $\lambda$ ends because the corresponding label
  $l_{\sigma}$ in the entry $\mu_{\sigma}$ gets canceled, but lasts a
  practically infinite long time, then this epoch can be considered,
  from an informal point of view, safe.  One could worry about having
  only very ``short'' epochs at $\lambda$ due to some inconsistencies
  (canceling labels, or entries with high values in the step and trial
  fields) in the system.  Theorem~\ref{th:existence_safe_epoch} shows
  that every time a ``short'' epoch ends,  the system somehow
  loses one of its inconsistencies, and, eventually, the proposer
  $\lambda$ reaches a practically infinite epoch.  Note also that a
  $0$-safe epoch and a $1$-safe or a $2$-safe epoch are, in practice,
  as long as each other. Indeed, any $h$-safe epoch with $h$ very
  small compared to $2^{\mathfrak{b}}$ can be considered practically
  infinite. Whether $h$ can be considered very small depends on the
  concrete timescale of the system.
\end{remark}

\begin{remark}
  Besides, every processor $\alpha$ always checks that the entry
  $\alpha$ is valid, and, if not, it produces a new label in the entry
  $a_{\alpha}[\alpha]$ and resets the step, trial and canceling label
  field. Doing so, even if $\alpha$'s first valid entry $\mu$ is located
  before the entry $\alpha$, the processor $\alpha$ still works to
  find a ``winning'' label for its entry $\alpha$. In that case, if
  the entry $\mu$ becomes invalid, then the entry $\alpha$ is ready to
  be used, and a safe epoch can start without waiting any longer.
\end{remark}

% \subsection{Practical Safety - Definitions}
% \label{app:practical_safety_definitions}

\subsection{Safety - Definitions}
\label{app:practical_safety_definitions}

To prove the safety property within a subexecution, we have to focus
on the events that correspond to deciding a proposal, e.g., $(b,p)$ at
processor $\alpha$.  Such an event may be due to corrupted messages in
the communication channels an any stage of the Paxos algorithm.
Indeed, a proposer selects the proposal it will send in its phase $2$
thanks to the replies it has received at the end of its phase
$1$. Hence, if one of these messages is corrupted, then the safety
might be violated. However, there is a finite number of corrupted
messages since the capacity of the communication channels is
finite. Hence, violations of the safety do not happen very often.  To
formally deal with these issues, we define the notion of scenario that
corresponds to specific chain of events involved in the Paxos
algorithm.

\begin{definition}[Scenario]
  \label{def:scenario}
  Consider a subexecution $E = (\gamma_k)_{k_0 \leq k \leq k_1}$.  A \emph{scenario
  in $E$} is a sequence $U = (U_i)_{0 \leq i < I}$ where each $U_i$ is a
  collection of events in $E$. In addition, every event in $U_i$
  happens before every event in $U_{i+1}$.
  We use the following notations
  \begin{itemize}
  \item $\rho \xrightarrow{p1a} (S,b)$ : The proposer $\rho$
    broadcasts a message $p1a$ containing the tag $b$. Every acceptor
    in the quorum $S$ receives this message and adopts\footnote{Recall
      that this means it copies the entry $b[\chi(b)]$ in the entry
      $a_{\beta}[\chi(b)]$.} the tag $b$.
  \item $(S,b) \xrightarrow{p1b} \rho$ : Every processor $\alpha$ in
    the quorum $S$ sends to the proposer $\rho$ a $p1b$ message
    telling they adopted the tag $b$, and containing the last proposal
    $r_{\alpha}[\chi(a_{\alpha})]$ they accepted. These messages are received
    by $\rho$.
  \item $\rho \xrightarrow{p2a} (Q,b,p)$ : The proposer $\rho$
    broadcasts a $p2a$ message containing a proposal $(b,p)$.  Every
    acceptor in the quorum $Q$ accepts the proposal $(b,p)$.
  \item $(Q,b,p) \xrightarrow{p2b} \rho$ : Every acceptor $\alpha$ in
    the quorum $Q$ sends to the proposer $\rho$ a $p2b$ message
    telling that it has accepted the proposal $(b,p)$. The proposer $\rho$
    receives these messages.
  \item $\rho \xrightarrow{dec} (\alpha,b,p)$ : the proposer $\rho$
    sends a decision message containing the proposal $(b,p)$.  The
    processor $\alpha$ receives this message, accepts and decides on
    the proposal $(b,p)$.
  \end{itemize}
\end{definition}

\begin{definition}[Simple Acceptation Scenario]
  \label{def:simple_acceptation_scenario}
  Given $S$ a quorum of
  acceptors, $b$ a tag, $p$ a consensus value, $\rho$ a proposer and
  $\alpha$ an acceptor, a \emph{simple acceptation scenario $U$ of the
    first kind} is defined as follows.
  \begin{itemize}
  \item[$(U_0)$] A proposer $\rho$ broadcasts a $p1a$ message with tag $b$.
  \item[$(U_1)$] Every processor $\beta$ from a quorum $S$ receives
    this $p1a$ message, adopts the tag $b$ and replies to
    $\rho$ a $p1b$ message containing its tag $a_{\beta} \simeq b$ and
    the lastly accepted proposal $r_{\beta}[\chi(a_{\beta})]$.
  \item[$(U_2)$] The proposer $\rho$ receives these messages at the
    end of its Paxos phase $1$, moves to the second phase of Paxos,
    and sends a $p2a$ message to a processor $\alpha$ telling it to
    accept the proposal $(b,p)$.
  \item[$(U_3)$] The processor $\alpha$ receives the $p2a$ message and
    accepts the proposal $(b,p)$.
  \end{itemize}

  Given quorums $S$ and $Q$,
  $b$ a tag, $p$ a consensus value, $\rho$ a proposer and $\alpha$ an
  acceptor, a \emph{simple acceptation scenario $V$ of the second
    kind} is defined as follows.
  \begin{itemize}
  \item[$(V_0)$] A proposer $\rho$ broadcasts a $p1a$ message with tag $b$.
  \item[$(V_1)$] Every processor $\beta$ from a quorum $S$ receives
    this $p1a$ message, adopts the tag $b$ and replies to
    $\rho$ a $p1b$ message containing its tag $a_{\beta} \simeq b$ and
    the lastly accepted proposal $r_{\beta}[\chi(a_{\beta})]$.
  \item[$(V_2)$] The proposer $\rho$ receives these messages at the
    end of its Paxos phase $1$, moves to the second phase of Paxos,
    and sends a $p2a$ message to every processor in $Q$ telling it to
    accept the proposal $(b,p)$.
  \item[$(V_3)$] Every processor in $Q$ receives the $p2a$ message,
    accepts the proposal and replies to the proposer $\rho$ with a
    $p2b$ message.
  \item[$(V_4)$] The proposer $\rho$ receives the replies from the
    acceptors in $Q$, and sends to the acceptor $\alpha$ a decision
    message containing a proposal $(b,p)$.
  \item[$(V_5)$] The acceptor $\alpha$ receives the decision message,
    accepts and decides on the proposal $(b,p)$.
  \end{itemize}

  With the notations introduced, we have 
  % \begin{equation}
  %   \textrm{(1-st kind) } \rho \xrightarrow{p1a} (S,b) \xrightarrow{p1b} \rho \xrightarrow{p2a} (\alpha,b,p)
  % \end{equation}
  % \begin{multline}
  %   \textrm{(2-nd kind) } \rho \xrightarrow{p1a} (S,b) \xrightarrow{p1b} \rho \xrightarrow{p2a}
  %   (Q,b,p) \xrightarrow{p2b} \rho \xrightarrow{dec} (\alpha,b,p)
  % \end{multline}
  \begin{align}
    \textrm{(1-st kind) } &\rho \xrightarrow{p1a} (S,b) \xrightarrow{p1b} \rho \xrightarrow{p2a} (\alpha,b,p)\\
    \textrm{(2-nd kind) } &\rho \xrightarrow{p1a} (S,b) \xrightarrow{p1b} \rho \xrightarrow{p2a}
    (Q,b,p) \xrightarrow{p2b} \rho \xrightarrow{dec} (\alpha,b,p)
  \end{align}
  If the kind of scenario is not relevant, we note $S \leadsto (\alpha,b,p)$.
  % \begin{equation}
  %   S \leadsto (\alpha,b,p)
  % \end{equation}
\end{definition}

\begin{remark}
  A simple acceptation scenario is simply a basic execution of the
  Paxos algorithm that leads a processor to either accept a proposal,
  or decide on a proposal (accepting it by the way).
\end{remark}

\begin{definition}[Fake Message]
  \label{def:fake_message}
  Given a subexecution $E = (\gamma_k)_{k_0 \leq k \leq k_1}$, a
  \emph{fake message relatively to the subexecution $E$}, or simply a
  fake message, is a message that is in the communication channels in
  the first configuration $\gamma_{k_0}$ of the subexecution $E$.
\end{definition}

\begin{remark}
  This definition of fake messages comprises the messages at the
  beginning of $E$ that were not sent by any processor, but also
  messages produced in the prefix of execution that precedes $E$.
\end{remark}

\begin{definition}[Simple Fake Acceptation Scenario]
  \label{def:simple_fake_acceptation_scenario}
  Given a subexecution $E$, we note $\bigcirc \rightarrow X$ if there
  exists an event $e$ in $X$ that corresponds to the reception of a
  fake message relatively to $E$.  With the previous notation, a
  \emph{simple fake acceptation scenario relatively to $E$} is one of the following
  scenario.
  \begin{align}
%    \bigcirc& \xrightarrow{p1a} (S,b) \\
    \bigcirc& \xrightarrow{p2a} (\alpha,b,p) \\
    \bigcirc& \xrightarrow{p1b} \rho \xrightarrow{p2a} (\alpha,b,p) \\
    \bigcirc& \xrightarrow{dec} (\alpha,b,p) \\
    \bigcirc& \xrightarrow{p2b} \rho \xrightarrow{dec} (\alpha,b,p) \\
    \bigcirc& \xrightarrow{p2a} (Q,b,p) \xrightarrow{p2b} \rho \xrightarrow{dec} (\alpha,b,p) \\
    \bigcirc& \xrightarrow{p1b} \rho \xrightarrow{p2a} (Q,b,p) \xrightarrow{p2b} \rho \xrightarrow{dec} (\alpha,b,p)
  \end{align}
  If the exact type is not relevant, we note $\bigcirc \leadsto
  (\alpha,b,p)$.
\end{definition}

\begin{remark}
  A simple fake acceptation scenario is somehow similar to a simple
  acceptation scenario except the fact that at least one fake message
  (relatively to the given subexecution) is involved during the
  scenario.
\end{remark}

\begin{definition}[Composition]
  \label{def:composition}
  Consider two simple scenarios $U = X \leadsto (\alpha_1,b_1,p_1)$,
  where $X = \bigcirc$ or $X = (S_1,b_1)$, and $V = S_2 \leadsto
  (\alpha_2,b_2,p_2)$ such that the following conditions are satisfied.
  \begin{itemize}
  \item The processor $\alpha_1$ belongs to $S_2$
  \item Let $e_2$ be the event that corresponds to $\alpha_1$ sending
    a $p1b$ message in scenario $V$.  Then the event ``$\alpha_1$
    accepts the proposal $(b_1,p_1)$'' is the last acceptation event
    before $e_2$. In addition, the proposer involved in the scenario
    $V$ selects the proposal $(b_1,p_1)$ as the highest-numbered
    proposal at the beginning of the Paxos phase $2$. In particular,
    $p_1 = p_2$.
  \item All the tags involved share the same first valid entry, the
    same corresponding label and step value.
  \end{itemize}
  Then the \emph{composition} of the two simple scenarios is the
  concatenation the scenarios $U$ and $V$. This scenario is noted
  \begin{equation}
    X \leadsto (\alpha_1,b_1,p_1) \rightarrow S_2 \leadsto (\alpha_2,b_2,p_2)
  \end{equation}
  Note that the trial value is strictly increasing along the simple
  scenarios.
\end{definition}

\begin{figure}
  \centering
  \includegraphics[scale=0.4]{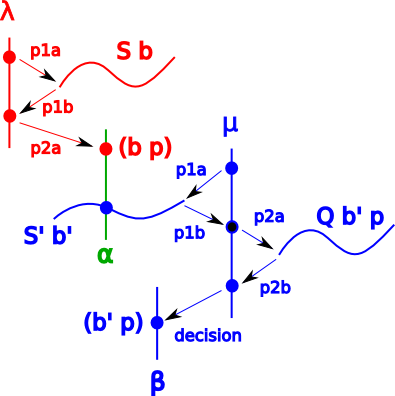}
  \caption{Composition of scenarios of the 1-st kind (red) and the
    2-nd kind (blue) - Time flows downward, straight lines are local
    executions, arrows represent messages.}
  \label{fig:composition_scenarios}
\end{figure}

\begin{definition}[Acceptation Scenario]
  \label{def:acceptation_scenario}
  Given a subexecution $E$, an \emph{acceptation scenario} is the
  composition $U$ of simple acceptation scenarios $U_1, \dots, U_r$
  where $U_1$ is either a simple acceptation scenario or a simple fake
  acceptation scenario relatively to $E$. We note
  \begin{equation} 
    X \leadsto (\alpha_1,b_1,p) \rightarrow S_2
    \leadsto (\alpha_2,b_2,p) \dots S_r \leadsto (\alpha_r,b_r,p)
  \end{equation} 
  An acceptation scenario whose first simple scenario is not fake relatively to $E$ is
  called \emph{real acceptation scenario relatively to $E$}.
  An acceptation scenario whose first simple scenario is fake relatively to $E$ is
  called \emph{fake acceptation scenario relatively to $E$}.
  Given an event $e$ that corresponds to some processor
  accepting a proposal, we note $Sc(e)$ the set of
  acceptation scenarios that ends with the event $e$.    
\end{definition}

\begin{remark}
  Given an acceptation event or a decision event, there is always at
  least one way to trace back the scenario that has lead to this event.
  If one of these scenarios involve a fake message, then we cannot
  control the safety property for the corresponding step.
  Besides, note that all the tags involved share the same first valid entry
  $\mu$, the same corresponding label $l$, step value $s$ and
  consensus value $p$.  Also, the trial value is increasing along the
  acceptation scenario.
\end{remark}

\begin{definition}[Scenario Characteristic]
  The \emph{characteristic} of an acceptation scenario $U$ in which all
  tags have first valid entry $\mu$, corresponding label $l$, step
  value $s$ and consensus value $p$, is the tuple $char(U) =
  (\mu,l,s,p)$.
\end{definition}

\begin{definition}[Fake Characteristics]
  Consider a subexecution $E = (\gamma_k)_{k_0 \leq k \leq k_1}$.
  Given a scenario characteristic $(\mu,l,s,p)$, we note $\mathcal{E}(E,\mu,l,s,p)$
  the set of events in $E$ that correspond to accepting a proposal $(b,p)$ with
  $\chi(b) = \mu$ and $b[\mu].(l~s) = (l~s)$.
  A characteristic $(\mu,l,s,p)$ is said to be \emph{fake relatively to $E$}
  if there exists an event $e$ in $\mathcal{E}(E,\mu,l,s,p)$ such that
  the set $Sc(e)$ contains a fake acceptation scenario relatively to $E$.
  We note $\mathcal{FC}(E)$ the set of fake characteristics relatively to $E$.  
\end{definition}

\begin{definition}[Unsafe Steps]
  If we fix the identifier $\mu$ and the label $l$, we define the set of
  \emph{unsafe step} values $\mathcal{US}(E,\mu,l)$ as the set of values $s$ such
  that there exists a consensus value $p$ with $(\mu,l,s,p) \in
  \mathcal{FC}(E)$. 
\end{definition}

\begin{remark}
  Given an identifier $\mu$ and the label $l$, an unsafe step $s$ is a
  step in which an accepted proposal might be induced by fake
  messages, and thus, we cannot control the safety for this step.
\end{remark}

\begin{definition}[Observed Zone]
  Consider an execution $E$.  Let $\lambda$ be a proposer and let
  $\Sigma$ be a subexecution such that the local execution $\sigma =
  \Sigma(\lambda)$ at $\lambda$ is a $h$-safe epoch. We note $F$ the
  suffix of the execution that starts with $\Sigma$.  Assume that
  $\lambda$ executes at least two trials during its epoch $\sigma$.
  Let $Q^0$, $Q^f$ be the first and last quorums respectively whose
  messages are processed by the proposer $\lambda$ during $\sigma$.
  For each processor $\alpha$ in $Q^0$ (resp. $Q^f$), we note
  $e^0(\alpha)$ (resp. $e^f(\alpha)$) the event that corresponds to
  $\alpha$ sending to $\lambda$ a message received in the trial that corresponds
  to $Q^0$ (resp. $Q^f$).

  The \emph{zone observed by $\lambda$ during the epoch $\sigma$},
  noted $Z(F,\lambda,\sigma)$, is the set of real acceptation scenarios
  relatively to $F$ described as follows. A real acceptation scenario
  relatively to $F$ belongs to $Z(F,\lambda,\sigma)$ if and only if it
  ends with an acceptation event that does not happen after the end of
  $\sigma$ and its first simple acceptation scenario $U = (S,b) \leadsto
  (\beta,b,p)$ is such that there exists an acceptor $\alpha$ in $S
  \cap Q^0 \cap Q^f$ at which the event $e^0(\alpha)$ happens before
  the event $e$ that corresponds to sending a $p1b$ message in $U$, and the
  event $e$ happens before the event $e^f(\alpha)$
  (cf. Figure~\ref{fig:practss_safezone}).
\end{definition}

\begin{figure}
  \centering
  \includegraphics[scale=0.4]{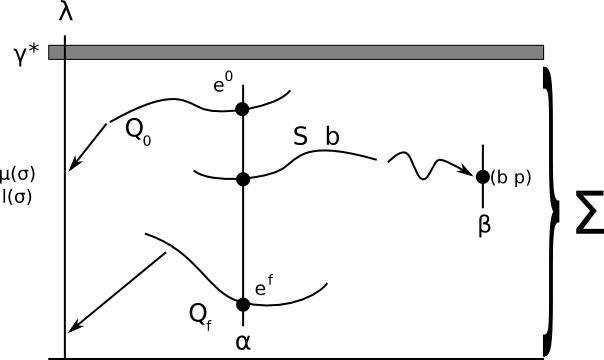}
  \caption{Scenario $S \leadsto (\beta,b,p)$ in $Z(F,\lambda,\sigma)$ - Time
    flows downward, straight lines are local executions, curves are
    send/receive events, arrows represent messages.}
  \label{fig:practss_safezone}
\end{figure}

\begin{remark}
  The observed zone models a globally defined time period during
  which we will prove, under specific assumptions, the safety property
  (cf. Theorem~\ref{th:practical_safety}).
\end{remark}

% \subsection{Practical Safety - Results}
% \label{app:practical_safety_results}
\subsection{Safety - Results}
\label{app:practical_safety_results}

\begin{lemma}[Fake Acceptation Scenarios]
  \label{lem:fake_acceptation_scenarios}
  Consider a fake message $m$, and two acceptation scenarios of
  charactestics $(\mu,l,s,p)$ and $(\mu',l',s',p')$ that begins with
  the reception of $m$.  Then both scenarios share the same
  characteristics, i.e., $(\mu,l,s,p) = (\mu',l',s',p')$.
\end{lemma}

\begin{proof}
  We have two scenarios that begins with the reception of $m$. Focus
  on the first simple scenario of each acceptation scenario.  Assume,
  for instance, that the message $m$ is a $p1b$ message and both
  simple fake acceptation scenarios are as follows
  \begin{align}
    \bigcirc& \xrightarrow{p1b} \rho \xrightarrow{p2a} (\alpha,b,p) \\
    \bigcirc& \xrightarrow{p1b} \rho' \xrightarrow{p2a} (\alpha',b',p')
  \end{align}
  Since once a message is received, it is not in the communication
  channels anymore, the event ``reception of $m$ at $\rho$'' and
  ``reception of $m$ at $\rho'$'' must be the same. In particular
  $\rho = \rho'$. Thanks to the messages it has received, the
  processor $\rho$ computes a proposal $(b,p)$ and broadcasts
  it. Hence, the processors $\alpha$ and $\alpha'$ receives (and accepts) the same
  proposal $(b,p)$.  Hence, $(b,p) = (b',p')$. By definition, $\chi(b)
  = \mu$, $b[\mu].(l~s) = (l~s)$ and $\chi(b') = \mu'$, $b[\mu'].(l~s)
  = (l'~s')$.  Therefore, $(\mu,l,s,p) = (\mu',l',s',p')$. The other
  cases are analogous.\gqed
\end{proof}

\begin{theorem}[Fake Characteristics]
  \label{th:fake_scenario_characteristic}
  Consider an execution $E$. Let $\lambda$ be a proposer, and let
  $\Sigma$ be a subexecution such that the local execution $\sigma =
  \Sigma(\lambda)$ at $\lambda$ is an $h$-safe epoch, with first valid entry $\mu_{\sigma}$ and label $l_{\sigma}$.
  We note $F$ the suffix of execution that starts with $\Sigma$.
  Then, for every fake characteristic $(\mu_{\sigma},l_{\sigma},s,p) \in \mathcal{FC}(F)$, we have $s < h$.
  In other words, every step $s \in \mathcal{US}(F,\mu_{\sigma},l_{\sigma})$ satisfies $s < h$.
\end{theorem}

\begin{proof}
  Let $\gamma^*$ denote the configuration right before $\Sigma$.
  Consider any fake scenario of characteristic
  $(\mu_{\sigma},l_{\sigma},s,p)$ relatively to $F$. The scenario
  begins by the reception of one or more fake messages. But, each of
  these fake messages carry tags with first valid entry $\mu_{\sigma}$
  and label $l_{\sigma}$ that were present in $\gamma^*$. Hence, the
  corresponding step fields must have values less than $h$.  \gqed
\end{proof}

% \begin{theorem}[Fake Characteristics]
%   \label{th:fake_scenario_characteristic}
%   Given a subexecution $F$, we have
%   \begin{align}
%     |\mathcal{FC}(F)| \leq \gconst{C}\frac{\gconst{n}(\gconst{n}-1)}{2} \\
%     \forall \mu,l, ~ |\mathcal{US}(F,\mu,l)| \leq \gconst{C}\frac{\gconst{n}(\gconst{n}-1)}{2}
%   \end{align}
% \end{theorem}

% \begin{proof}
%   On the contrary, note $r = |\mathcal{FC}(F)|$ and assume that $r >
%   \gconst{C}\frac{\gconst{n}(\gconst{n}-1)}{2}$.  We note
%   $c_1,\dots,c_r$ the distinct values in $\mathcal{FC}(F)$, and $c_i =
%   (\mu_i,l_i,s_i,p_i)$.  By definition, for each $c_i$, there is an
%   event $e_i \in \mathcal{FC}(F,\mu_i,l_i,s_i)$ and a scenario $U_i$ in
%   $Sc(e_i)$ that begins with a simple fake acceptation scenario. Let
%   $m_i$ be the fake message consumed in this simple fake acceptation
%   scenario.  Since there are at most
%   $\gconst{C}\frac{\gconst{n}(\gconst{n}-1)}{2}$ fake messages in the
%   starting configuration of $F$, and since the set of
%   fake messages is non-increasing (channels do not produce messages),
%   there are two indices $1 \leq i < j \leq r$ such that $m_i =
%   m_j$. By Lemma~\ref{lem:fake_acceptation_scenarios}, the scenarios
%   $U_i$ and $U_j$ have the same characteristic, i.e., $c_i = c_j$;
%   contradiction.  The second inequation is straightforward since
%   $|\mathcal{US}(F,\mu,l)| \leq |\mathcal{FC}(F)|$\gqed
% \end{proof}

\begin{lemma}[Epoch and Cycle of Labels]
  \label{lem:epoch_cycle_labels}
  Consider an execution $E$.  Let $\lambda$ be a processor and
  consider a subexecution $\Sigma$ such that the local execution
  $\sigma = \Sigma(\lambda)$ is an epoch at $\lambda$.  We note $F$
  the suffix of the execution $E$ that starts with $\Sigma$. Consider
  a processor $\rho$ and a finite subexecution $G$ in $F$ as follows:
  $G$ starts in $\Sigma$ and induces a local execution $G(\rho)$ at
  $\rho$ such that it starts and ends with the first valid entry of
  the tag $a_{\rho}$ being equal to $\mu_{\sigma}$ and containing the
  label $l_{\sigma}$, and the label field in the entry
  $a_{\rho}[\mu_{\sigma}]$ undergoes a cycle of labels during
  $G(\rho)$.  Assume that, if $\mu_{\sigma} < \lambda$, then the
  processor $\mu_{\sigma}$ does not produce any label during $G$.
  Then $\mu_{\sigma} = \lambda$ and the last event of $\sigma$ happens
  before the last event of $G(\rho)$.
\end{lemma}

\begin{proof}
  By Lemma~\ref{lem:cycle_labels}, since the entry $a_{\rho}[\lambda]$
  remains valid after the readoption of the label $l$ at the end of
  $G(\rho)$, the proposer $\mu_{\sigma}$ must have produced some label
  $l'$ during $G$ (hence $\mu_{\sigma} = \lambda$) that was received
  by $\rho$ during $G$. Necessarily, the production of $l'$ happens
  after the last event of $\sigma$ at $\lambda$, thus the last event
  of $G(\rho)$ at $\rho$ also happens after the last event of $\sigma$
  at $\lambda$.\gqed
\end{proof}

\begin{theorem}[Weak Safety]
  \label{th:weak_practical_safety}
  Consider an execution $E$.  Let $\lambda$ be a processor and let
  $\Sigma$ be a subexecution such that the local execution $\sigma =
  \Sigma(\lambda)$ at $\lambda$ is an $h$-safe. We note $F$ the suffix
  of the execution that starts with $\Sigma$.  Consider a step value
  $s$ and the two following simple scenarios
  \begin{align}
    U_1 &= \rho_1 \xrightarrow{p1a} (S_1,b_1) \xrightarrow{p1b} \rho_1 \xrightarrow{p2a}
      (Q_1,b_1,p_1) \xrightarrow{p2b} \rho_1 \xrightarrow{dec} (\alpha_1,b_1,p_1)\\
    U_2 &= (S_2,b_2) \leadsto (\alpha_2,b_2,p_2)    
  \end{align}
  with characteristics $(\mu_{\sigma},l_{\sigma},s,p_1)$ and
  $(\mu_{\sigma},l_{\sigma},s,p_2)$ respectively.  In addition, we
  assume that $b_i[\mu_{\sigma}].t > h$ and $\tau_1 \leq \tau_2$ where
  $\tau_i = b_i[\lambda].(t~id)$.  We note $e_i$ for the acceptation
  event $(\alpha_i,b_i,p_i)$.  Assume that both events $e_1$ and $e_2$
  occur in $F$ and $s \not\in \mathcal{US}(F,\mu_{\sigma},l_{\sigma})$. In addition, assume
  that, if $\mu_{\sigma} < \lambda$, then the processor $\mu_{\sigma}$
  does not produce any label during $F$.  Then either $p_1 = p_2$ or
  the last event of $\sigma$ happens before one of the event $e_1$ or
  $e_2$.
\end{theorem}

\begin{proof}
  We assume that both events $e_1$ and $e_2$ do not happen after the
  last event of $\sigma$ and we prove that $p_1 = p_2$.  Since $s$ is
  not in $\mathcal{US}(F,\mu_{\sigma},l_{\sigma})$, every scenario in $Sc(e_1)$
  or $Sc(e_2)$ are real acceptation scenarios relatively to $F$.  We
  note $\gamma^*$ the configuration right before the subexecution
  $\Sigma$. We prove the result by induction on the value of $\tau_2$.

  \paragraph{\bf (Bootstrapping)}
  We first assume that $\tau_2 = \tau_1$. In particular, $\rho_1 =
  \tau_1.id = \tau_2.id = \rho_2$.  If $p_1 \neq p_2$, this means that
  $\rho_1$ has sent two $p2a$ messages with different proposals and
  the same tag\footnote{Modulo $\simeq$.}.  Note $e$ and $f$ the
  events that correspond to these two sendings. None of the events $e$
  and $f$ occurs in the execution prefix $A$, otherwise, since $e_1$
  and $e_2$ occur in $F$, the configuration $\gamma^*$ would contain a
  tag $x$ with $x[\mu_{\sigma}].l = l_{\sigma}$ and $x[\mu_{\sigma}].t
  > h$; this is a contradiction since $\sigma$ is $h$-safe.  Hence,
  $e$ and $f$ occur in $F$.  Then, there must be a cycle of labels in the
  entry $a_{\rho_1}[\mu_{\sigma}]$ between the $e$ and $f$.  By
  Lemma~\ref{lem:epoch_cycle_labels}, this implies that the last event
  of $\sigma$ happens before the event $e_1$ or $e_2$; this is a
  contradiction.  Hence, $p_1 = p_2$.

  \paragraph{\bf (Induction)} 
  Now, $\tau_2$ is any value such that $\tau_1 < \tau_2$ and we assume
  the result holds for every value $\tau$ such that $\tau_1 \leq \tau
  < \tau_2$.  Pick some acceptor $\beta$ in $Q_1 \cap S_2$.  From its
  point of view, there are two events $f_1$ and $f_2$ at $\beta$ that
  respectively correspond to the acceptation of the proposal
  $(b_1,p_1)$ in the scenario $U_1$ (reception of a $p2a$ message),
  and the adoption of the tag $b_2$ in the scenario $U_2$ (reception
  of a $p1a$ message). First, the events $f_1$ and $f_2$ do not occur
  in the execution prefix $A$.  Otherwise there would exist a tag
  value $x$ in $\gamma^*$ such that $x[\mu_{\sigma}].l = l_{\sigma}$
  and $x[\mu_{\sigma}].t > h$; this is a contradiction, since $\sigma$
  is $h$-safe. Hence, $f_1$ and $f_2$ occur in the suffix $F$.

  We claim that $f_1$ happens before $f_2$.  Otherwise, since $\tau_2
  > \tau_1$, there must be a cycle of labels in the field
  $a_{\beta}[\mu_{\sigma}].l$. By Lemma~\ref{lem:epoch_cycle_labels}, this
  implies that the last event of $\sigma$ happens before the event
  $f_1$, and thus before the event $e_1$; contradiction. Hence, $f_1$
  happens before $f_2$.  We claim that the $p1b$ message the acceptor
  $\beta$ has sent contains a non-null lastly accepted proposal
  $r_{\beta}[\mu_{\sigma}] = (b,p)$ such that $\chi(b) = \mu_{\sigma}$,
  $b[\mu_{\sigma}].(l~s) = (l_{\sigma}~s)$ and $\tau_1 \leq b[\mu_{\sigma}].(t~id) <
  \tau_2$.  Otherwise, there must be a cycle of labels between $f_1$
  and $f_2$, which implies that $f_2$, and thus $e_2$, happens after
  the end of $\sigma$.

  Now, the proposer $\rho_2$ receives a set of proposals from the
  acceptors of the quorum $S_2$, including at least one non-null
  proposal from $\beta$.  It first checks that every tag received uses
  the entry $\mu_{\sigma}$ and the label $l_{\sigma}$ and that there is no
  two different proposals with two tags that share the same content in
  entry $\mu_{\sigma}$ before continuing to the second phase of Paxos, and
  if it is not the case, it updates its proposer tag and executes
  another phase $1$ of Paxos.  Hence, since $\rho_2$ has moved to the
  second phase of Paxos, it means that no such issue has happened.
  Then, it selects among the proposals whose tags point to the step
  $s$ the proposal $(b_c,p_c)$ with the highest tag. In particular,
  $\chi(b_c) = \mu_{\sigma}$, $b_c[\mu_{\sigma}].(l~s) =
  (l_{\sigma}~s)$.  Since $\rho_2$ has received the proposal $(b,p)$
  from $\beta$, we have $\tau_1 \leq \tau_c < \tau_2$, where $\tau_c =
  \beta_c[\mu_{\sigma}].(t~id)$.  Let $\beta_c$ be the proposer in
  $S_2$ which has sent to $\rho_2$ the proposal $(b_c,p_c)$ in the
  $p1b$ message.  There is an event $f_c$ in $F$ that corresponds to
  $\beta_c$ accepting the proposal $(b_c,p_c)$. Otherwise there would
  exist a tag value $x$ in $\gamma^*$ such that $x[\mu_{\sigma}].l =
  l_{\sigma}$ and $x[\mu_{\sigma}].t > h$; this is a contradiction,
  since $\sigma$ is $h$-safe.  Next, since $s \not\in
  \mathcal{US}(F,\mu_{\sigma},l_{\sigma})$, $\chi(b_c) =
  \mu_{\sigma}$, and $b_c[\mu_{\sigma}].(l~s) = (l_{\sigma}~s)$,
  the set $Sc(e_2)$ does not contain any fake acceptation
  scenario relatively to $F$, thus neither the set $Sc(f_c)$.  We can
  pick a real scenario in $Sc(f_c)$ and apply the induction
  hypothesis, which shows that $p_c = p_1$.  Hence, $p_1 = p_2$, since $p_c$ is
  the consensus value the proposer $\rho_2$ sends during the
  corresponding Paxos phase $2$.\gqed
\end{proof}

\begin{corollary}[Weak Safety]
  \label{corol:weak_practical_safety}
  Consider an execution $E$.  Let $\lambda$ be a processor and let
  $\Sigma$ be a subexecution such that the local execution $\sigma =
  \Sigma(\lambda)$ at $\lambda$ is an $h$-safe epoch. We note $F$ the
  suffix of the execution that starts with $\Sigma$.  Consider a step
  value $s$ and two decision events $e_i = (\alpha_i,b_i,p_i)$, $i =
  1,2$, such that $\chi(b_i)=\mu_{\sigma}$, $b_i[\mu_{\sigma}].(l~s) =
  (l_{\sigma}~s)$ and $b_i[\mu_{\sigma}].t > h$.  Assume that both
  events $e_1$ and $e_2$ occur in $F$ and $s \geq h$.
  In addition, assume that, if $\mu_{\sigma} < \lambda$, then the
  processor $\mu_{\sigma}$ does not produce any label during $F$.
  Then either $p_1 = p_2$ or the last event of $\sigma$ happens before
  one of the event $e_1$ or $e_2$.
\end{corollary}

\begin{proof}
  Since $e_1$ and $e_2$ are decision events, and since $s$ is not in
  $\mathcal{US}(F,\mu_{\sigma},l_{\sigma})$ ($s \geq h$, cf
  Theorem~\ref{th:fake_scenario_characteristic}), there are two real
acceptation scenarios in $Sc(e_1)$ and $Sc(e_2)$ relatively to $F$
respectively that contains simple acceptation scenarios of the second
kind as follows:
  \begin{align}
    U_1 &= \rho_1 \xrightarrow{p1a} (S_1,c_1) \xrightarrow{p1b} \rho_1 \xrightarrow{p2a}
      (Q_1,c_1,p_1) \xrightarrow{p2b} \rho_1 \xrightarrow{dec} (\beta_1,c_1,p_1)\\
    U_2 &= \rho_2 \xrightarrow{p1a} (S_2,c_2) \xrightarrow{p1b} \rho_2 \xrightarrow{p2a}
      (Q_2,c_2,p_2) \xrightarrow{p2b} \rho_2 \xrightarrow{dec} (\beta_2,c_2,p_2)    
  \end{align}
  with characteristics $(\mu_{\sigma},l_{\sigma},s,p_1)$ and
  $(\mu_{\sigma},l_{\sigma},s,p_2)$ respectively and trial values
  $c_i[\mu_{\sigma}].t$ greater than $h$.  We note $\tau_i =
  c_i[\mu_{\sigma}].(t~id)$.  Whether $\tau_1 \leq \tau_2$ or $\tau_2
  \leq \tau_1$, Theorem~\ref{th:weak_practical_safety} yields the
  result.\gqed
\end{proof}

\begin{theorem}[Safety]
  \label{th:practical_safety}
  Consider an execution $E$, a proposer $\lambda$ proposer and a
  subexecution $\Sigma$ such that the local execution $\sigma =
  \Sigma(\lambda)$ at $\lambda$ is a $h$-safe epoch for some bounded
  integer $h$.  We note $F$ the suffix of execution that starts with
  $\Sigma$. Assume that the observed zone $Z(F,\lambda,\sigma)$ is
  defined and that, if $\mu_{\sigma} < \lambda$, then the processor
  $\mu_{\sigma}$ does not produce any label during $F$.  Consider two
  scenarios $U_1$, $U_2$ in $Z(F,\lambda,\sigma)$ with characteristics
  $(\mu_1,l_1,s_1,p_1)$ and $(\mu_2,l_2,s_2,p_2)$ such that
  $\mu_{\sigma} \leq \min(\mu_1,\mu_2)$ and both scenarios contain
  simple acceptation scenarios with tags whose associated trial values
  are greater than $h$.  Then $(\mu_1,l_1) = (\mu_2,l_2) =
  (\mu_{\sigma},l_{\sigma})$, and if $s_1 = s_2 \geq h$ then $p_1 = p_2$.
\end{theorem}

\begin{proof}
  Assume that the scenario $U_1$ is such that $\mu_1 >
  \mu_{\sigma}$. Let $V = (S,b) \leadsto (\beta,b,p)$ be its first
  simple acceptation scenario.  By definition of the observed zone
  $Z(F,\lambda,\sigma)$, there exists an acceptor $\alpha$ in $S \cap
  Q^0 \cap Q^f$ such that we have the happen-before relations
  $e^0(\alpha) \leadsto e \leadsto e^f(\alpha)$, where $e$ is the
  event that corresponds to $\alpha$ sending a $p1b$ message in the
  scenario $V$.  At $e^0(\alpha)$ and $e^f(\alpha)$, messages are sent
  to $\lambda$ and are processed during $\sigma$. Hence, the
  corresponding tag values of the variable $a_{\alpha}$ must use the
  entry $\mu_{\sigma}$ and the label $l_{\sigma}$. Otherwise, the
  message either is not processed or causes an interrupt at processor
  $\lambda$.  Now, at event $e$, the first valid entry of the variable
  $a_{\alpha}$ is $\mu_1 > \mu_{\sigma}$ which implies that the entry
  $\mu_{\sigma}$ is invalid.  Hence, between $e^0(\alpha)$ and
  $e^f(\alpha)$, the entry $a_{\alpha}[\mu_{\sigma}]$ becomes invalid
  and valid again.  There must be a cycle of labels in the label field
  $a_{\alpha}[\lambda].l$.  Lemma~\ref{lem:epoch_cycle_labels} implies
  that the last event of $\sigma$ happens before $e^f(\alpha)$; by the
  definition of $e^f(\alpha)$, this is a contradiction. Therefore
  $\mu_1 = \mu_{\sigma}$. If $l_1 \neq l_{\sigma}$, then there must
  also be a cycle of labels in the entry $a_{\alpha}[\mu_{\sigma}]$
  between $e^0(\alpha)$ and $e^f(\alpha)$, which leads to a
  contradiction again, thanks to the same argument.  Therefore, $l_1 =
  l_{\sigma}$. Of course, the previous demonstration also shows that
  $(\mu_2,l_2) = (\mu_{\sigma},l_{\sigma})$.  If $s_1 = s_2 \geq h$, then
  Corollary~\ref{corol:weak_practical_safety}, the fact that the trial
  values associated to the scenarios $U_1$ and $U_2$ are greater than
  $h$ and the fact that the two acceptation events in scenarios $U_1$
  and $U_2$ do not happen after the end of $\sigma$ imply that $p_1 =
  p_2$.\gqed
\end{proof}

\begin{remark}
  In the case $\mu_{\sigma} < \lambda$ , assuming that $\mu_{\sigma}$
  does not produce any label during $F$ means that the proposer
  $\lambda$ should be the live processor with the lowest identifier.
  To deal with this issue, one can use a failure detector.
\end{remark}

\section{Generalized Self-Stabilizing Paxos}
\label{app:generalized_paxos}

A \emph{command history} is a sequence $c_1 c_2 \dots c_r$ of
state-machine commands of length at most $r < 2^{\mathfrak{b}}$.  Given
sequences $p_1$ and $p_2$, we note$p_1 \sqsubseteq p_2$ when $p_1$ is
a prefix of $p_2$.  Each learner $\alpha$, has a command history
variable $learned_{\alpha}$ that represents the sequence of commands
decided so far. We denote the empty sequence by $\bot$, the
concatenation of two command histories by $p_1 \circ p_2$ and the
length of a command history $p$ by $|p|$. Following
\cite{Lamport05generalizedconsensus}, the Generalized Consensus
specifications are:
\begin{itemize}
\item (Non-triviality) For any learner $\alpha$, $learner_{\alpha}$ is
  always a sequence of proposed commands.
\item (Stability) For any learner $\alpha$, the value of
  $learner_{\alpha}$ at any time is a prefix of its value at any later
  time.
\item (Consistency) For any learners $\alpha$ and $\beta$, it is
  always the case that one of the sequences $learned_{\alpha}$ and
  $learned_{\beta}$ is a prefix of the other.
\item (Liveness) If command $cmd$ is proposed and $\alpha$ is a
  learner, then eventually $learned_{\alpha}$ will contain the command
  $cmd$.
\end{itemize}

\subsection{Algorithm}

Our generalized self-stabilizing Paxos, assigns empty history of
state-machine commands whenever an epoch change takes place.
Thus, when the epoch is not changed for practically infinite execution
the accumulated history of state-machine commands
is extended by practically infinitely many new globally decided upon
state-machine commands, and therefore act
as as a virtual single state machine. Such an execution allows the
replicated state machine to stabilize in the interaction
with the (possibly interactive) algorithms that use the replicated
state machine as their virtual machine.

We now explain how to adapt the self-stabilizing repeated consensus algorithm to obtain a
Generalized Self-Stabilizing Paxos. We choose the type
of consensus value to be a command history. We keep the same variables
as before, and we simply add command history $learned_{\alpha}$ that
is modified only on decisions (Algorithm~\ref{alg:gpaxos_decide}).  The
acceptor algorithm and the preempting routine are not modified.  

We only add minor modifications to the proposer algorithm
(Algorithm~\ref{alg:gpaxos_proposer}) as follows. At the beginning of
the loop (line~\ref{algline:gpaxos_loop_read_proposal}), the proposer,
say $\lambda$, reads a command $cmd$, and initializes its variable
$p_{\lambda}$ to $p^* \circ cmd$ (i.e. the command history with a
single command) where $p^*$ is the value of $p_{\lambda}$ at the
beginning of the loop. The phase $1$ remains the same. At the
beginning of phase $2$, the proposer $\lambda$ has collected replies
(set $\Gamma$, line~\ref{algline:gpaxos_loop_p2_begin}) from a
majority of acceptors.
If the tags in $\Gamma$ satisfies a coherence condition, the proposer
selects the proposals $(a,p)$ such that the tag $a$ is maximal, and
discards those that do not satisfy $a[\chi(a_{\lambda}].s =
|p|$. Noting $\Gamma_0$ the filtered proposals, $\lambda$ selects the
command sequence $p_{max}$ in $\Gamma_0$ that is maximal according to
some lexicographical order (to break ties), and $\lambda$ sets its
command history $p_{\lambda}$ to the concatenation $p_{max} \circ cmd$. Note
that, if there are only null collected replies, or if all of the
replies were incoherent, then $p_{\lambda}$ keeps its value $p^*
\circ cmd$. Next, the proposer $\lambda$ executes the second phase as
in the previous algorithm.

In addition, any time the proposer undergoes a change of first valid
entry, or a change of label, the proposer either cuts its command
sequence $p_{\lambda}$ (via $p^*$ in the pseudo-code) or fill it with
$nop$ operations in order to have a length equal to the step field in
the first valid entry of $a_{\lambda}$ (command \texttt{truncate}).

\subsection{Proofs}

From the four requirements of Generalized Consensus, we only outline
the proofs for stability and consistency requirements.  Indeed, the
non-triviality condition follows from the fact that histories are
extended only by processors, namely by a concatenation of a new
command to existing histories (Algorithm~\ref{alg:gpaxos_proposer},
lines~\ref{algline:gpaxos_loop_append1} and
\ref{algline:gpaxos_loop_append2}).  And the liveness condition relies
on a failure detector, or more precisely, on the possibility for a
proposer to complete the Paxos phases; which is common to the Paxos
algorithm.

Theorem~\ref{th:existence_safe_epoch}, that ensures the existence
of a safe epoch, is still valid in this framework since the tag system
is not related to the type of consensus values and to the way they are
processed. Theorem~\ref{th:practical_safety} can be reformulated as follows.
\begin{theorem}[Generalized Paxos Stability and Consistency]
  \label{th:gpaxos_practical_safety}
  Consider an execution $E$, a proposer $\lambda$ proposer and a
  subexecution $\Sigma$ such that the local execution $\sigma =
  \Sigma(\lambda)$ at $\lambda$ is a $h$-safe epoch for some bounded
  integer $h$.  We note $F$ the suffix of execution that starts with
  $\Sigma$. Assume that the observed zone $Z(F,\lambda,\sigma)$ is
  defined and that, if $\mu_{\sigma} < \lambda$, then the processor
  $\mu_{\sigma}$ does not produce any label during $F$.  Consider two
  scenarios $U_1$, $U_2$ in $Z(F,\lambda,\sigma)$ with characteristics
  $(\mu_1,l_1,s_1,p_1)$ and $(\mu_2,l_2,s_2,p_2)$ such that
  $\mu_{\sigma} \leq \min(\mu_1,\mu_2)$ and both scenarios contain
  simple acceptation scenarios with tags whose associated trial values
  are greater than $h$.  Then $(\mu_1,l_1) = (\mu_2,l_2) =
  (\mu_{\sigma},l_{\sigma})$, and if $h \leq s_1 \leq s_2$ then $p_1 \sqsubseteq p_2$.
\end{theorem}
This theorem ensures the stability and consistency condition of the Generalized Consensus problem.

\twocolumn

\section{Algorithm Pseudo-Code}
\label{app:algorithm_pseudo_code}

% \subsection{Repeated Consensus}
% \label{subsec:repeatedconsensus}

\begin{algorithm}
  \footnotesize
  \SetKwBlock{KwFunction}{function}{end function}
  \SetKwFunction{KwAuxFillCl}{aux\_fill\_cl}
  \SetKwFunction{KwFillCl}{fill\_cl}
  \SetKwFunction{KwClean}{clean}
  \KwFunction({\KwClean{$\lambda$ : processor identifier, $a$ : tag}}){
    \ForEach{$\mu \in \Pi$}{
      \lIf{$a[\mu].cl \preccurlyeq a[\mu].l$}{$a[\mu] \leftarrow \bot$}\;
      $a[\mu].id \leftarrow \lambda$\;
    }
  }
  \KwFunction({\KwFillCl{$x,y$ : tags}}){
    $x_c \leftarrow x$, $y_c \leftarrow y$\;
    \ForEach{$\mu \in \Pi$}{
      \lIf{$y_c[\mu].(l~or~cl) \not\preccurlyeq x[\mu].l$}{$x[\mu].cl \leftarrow y_c[\mu].(l~or~cl)$}\;
      \uIf{$y_c[\mu].l = x[\mu].l \land y_c[\mu].(s~or~t) = 2^{\mathfrak{b}}$}{
        $x[\mu].(s~t) \leftarrow (2^{\mathfrak{b}} ~ 2^{\mathfrak{b}})$\;
      }
      idem by exchanging $(x_,x_c)$ and $(y,y_c)$
    }
  }
  \caption{Tags - Procedures}
  \label{alg:tagprocs2}
\end{algorithm}

\begin{algorithm}
  \footnotesize
  \SetKwBlock{KwFunction}{function}{end function}
  \SetKwFunction{KwFillCl}{fill\_cl}
  \SetKwFunction{KwClean}{clean}
  \SetKwFunction{KwCheckEntry}{check\_entry}
  \KwFunction(\KwCheckEntry{$\lambda$ : identifier, $x$ : tag, $L$ : history of labels}){
    \uIf{$x[\lambda]$ is invalid}{
      $L \leftarrow L + x[\lambda].l$\;
      $x[\lambda].(l~s~t~id) \leftarrow (\nu(L)~0~0~\lambda)$\;
      $x[\lambda].cl \leftarrow \bot$\;
    }
  }
  \KwFunction({$\nu^*( \lambda$ : identifier, $x$ : tag, $L$ : label history$)$}){
    $y \leftarrow x$\;
    \KwClean{$\lambda,y$}\;
    \uIf{$\chi(y) \leq \lambda$}{
      (case $\nu^s$) $y[\chi(y)].(s~t) \leftarrow (1+y[\chi(y)].s ~ 0)$\;
      (case $\nu^t$) $y[\chi(y)].t \leftarrow 1+y[\chi(y)].t$\;
    }
    \KwCheckEntry{$\lambda,y,L$}\;
    \Return{$y$}\;
  }
  \caption{Tags - Increment functions}
  \label{alg:tagprocs3}
\end{algorithm}

\begin{algorithm}
  \footnotesize
  % \DontPrintSemicolon
  \SetKwBlock{KwLoop}{loop}{end loop}
  \SetKwBlock{KwFunction}{function}{end function}
  \SetKwBlock{KwVariables}{variables}{}
  \SetKwBlock{KwRule}{rule}{end~rule}
  \SetKwFunction{KwSend}{send}
  \SetKwFunction{KwReceive}{receive}
  \SetKwFunction{KwBcast}{broadcast}
  \SetKwFunction{KwPropagate}{propagate}
  \SetKwFunction{KwDecide}{decide}
  \SetKwFunction{KwFillCl}{fill\_cl}
  \SetKwFunction{KwInput}{input}
  \SetKwFunction{KwCheckEntry}{check\_entry}
  \Switch{\KwReceive{}}{
    \uCase{$\langle p1a, \lambda, b \rangle$}{
      $a_{old} \leftarrow a_{\alpha}$\;
      \lIf{$b[\alpha].(l~or~cl) \not\preccurlyeq a_{\alpha}[\alpha].l$}
      {$H^{cl}_{\alpha} \leftarrow H^{cl}_{\alpha} + b[\alpha].(l~or~cl)$}\label{algline:acceptor_p1a_update_hcl}\;
      \KwFillCl{$a_{\alpha},b$}, \KwCheckEntry{$\alpha,a_{\alpha},H_{\alpha}^{cl}$}\label{algline:acceptor_p1a_fill_cl}\;
      \If{$a_{\alpha} \prec b$}{
        $a_{\alpha}[\chi(b)] \leftarrow b[\chi(b)]$\label{algline:acceptor_p1a_adopt}\;
        \uIf{$a_{old}[\chi(b)].l \neq a_{\alpha}[\chi(b)].l$}{
          $r_{\alpha}[\chi(b)] \leftarrow \bot$\label{algline:acceptor_p1a_label_change1}\;
          $H_{\alpha}[\chi(b)] \leftarrow H_{\alpha}[\chi(b)] + a_{old}[\chi(b)].l$\;
          \lIf{$\exists l \in H_{\alpha}[\chi(b)],~ l \not\preccurlyeq a_{\alpha}[\chi(b)].l$}
          {$a_{\alpha}[\chi(b)].cl \leftarrow l$}\label{algline:acceptor_p1a_label_change2}\;
        }
      }
      \ForEach{$\mu \in \Pi$}{
        $c \leftarrow r_{\alpha}[\mu].b$\;
        \uIf{$c[\mu].l \neq a_{\alpha}[\mu].l \lor a_{\alpha}[\mu].(l~s~t~id) \prec  c[\mu].(l~s~t~id)$}{
          $r_{\alpha}[\mu] \leftarrow \bot$\;
        }
      }
      \KwSend{$\lambda, \langle p1b, \alpha, a_{\alpha}, r_{\alpha}[\chi(a_{\alpha})] \rangle$}\label{algline:acceptor_p1a_reply}\;
    }
    \uCase{$\langle p2a, \lambda, b, p \rangle$ or $\langle decision, \lambda, b, p \rangle$}{
      $a_{old} \leftarrow a_{\alpha}$\;
      \lIf{$b[\alpha].(l~or~cl) \not\preccurlyeq a_{\alpha}[\alpha].l$}
      {$H^{cl}_{\alpha} \leftarrow H^{cl}_{\alpha} + b[\alpha].(l~or~cl)$}\label{algline:acceptor_p2a_update_hcl}\;
      \KwFillCl{$a_{\alpha},b$}, \KwCheckEntry{$\alpha,a_{\alpha},H_{\alpha}^{cl}$}\label{algline:acceptor_p2a_fill_cl}\;
      \If{$a_{\alpha} \preccurlyeq b$}{
        $a_{\alpha}[\chi(b)] \leftarrow b[\chi(b)]$,
        $r_{\alpha}[\chi(b)] \leftarrow [b, p]$\label{algline:acceptor_p2a_accept}\;
        \lIf{it is a decision message}{\KwDecide{$b,p$}}\label{algline:acceptor_decide}\;
        \uIf{$a_{old}[\chi(b)].l \neq a_{\alpha}[\chi(b)].l$}{
          $H_{\alpha}[\chi(b)] \leftarrow H_{\alpha}[\chi(b)] + a_{old}[\chi(b)].l$\label{algline:acceptor_p2a_label_change1}\;
          \lIf{$\exists l \in H_{\alpha}[\chi(b)],~ l \not\preccurlyeq a_{\alpha}[\chi(b)].l$}
          {$a_{\alpha}[\chi(b)].cl \leftarrow l$}\label{algline:acceptor_p2a_label_change2}\;
        }
      }
      \ForEach{$\mu \in \Pi$}{
        $c \leftarrow r_{\alpha}[\mu].b$\;
        \uIf{$c[\mu].l \neq a_{\alpha}[\mu].l \lor a_{\alpha}[\mu].(l~s~t~id) \prec  c[\mu].(l~s~t~id)$}{
          $r_{\alpha}[\mu] \leftarrow \bot$\;
        }
      }
      \lIf{it is a $p2a$ message}{\KwSend{$\lambda, \langle p2b, \alpha, a_{\alpha}, r_{\alpha} \rangle$}}\label{algline:acceptor_p2a_reply}\;
    }
  }    
  \caption{Acceptor $\alpha$}
  \label{alg:acceptor}
\end{algorithm}

\begin{algorithm}
  \footnotesize
  % \DontPrintSemicolon
  \SetKwBlock{KwLoop}{loop}{end~loop}
  \SetKwBlock{KwBlock}{}{}
  \SetKwBlock{KwRule}{rule }{end~rule}
  \SetKwBlock{KwFunction}{function}{end~function}
  \SetKwBlock{KwVariables}{variables}{}
  \SetKwFunction{KwSend}{send}
  \SetKwFunction{KwBcast}{broadcast}
  \SetKwFunction{KwFillCl}{fill\_cl}
  \SetKwFunction{KwUpdate}{update}
  \SetKwFunction{KwInput}{input}
  \SetKwFunction{KwExec}{execute}
  \SetKwFunction{KwPR}{PR}

  \KwLoop({As long as $\Theta_{\lambda} = \gconst{true}$}){
    $p^* \leftarrow \KwInput()$\label{algline:loop_read_proposal}\;
    $a_{\lambda} \leftarrow \nu^s(\lambda,a_{\lambda},H_{\lambda}^{cl})$\;
    \;
    [Ph. 1]\;
    $p_{\lambda} \leftarrow p^*$\;
    $\forall \alpha \in \Pi$, \KwSend{$\alpha, \langle p1a, \lambda, a_{\lambda} \rangle$}\label{algline:loop_bcast_p1a}\;
    \lIf{ \KwPR{$1$} returns $nok$}{go to [Ph. 1]}\;
    \;
    [Ph. 2]\;
    let $\mu = \chi(a_{\lambda})$, and $\Gamma$ be the set of non-null proposals $r_{\alpha}[\mu]$\label{algline:loop_p2_begin}
    received at the end of [Ph. 1] in\;
    \uIf{$\Gamma \neq \emptyset$}{
      \uIf{$\forall x,y \in \Gamma, \chi(x.a) = \chi(y.a) = \mu \land x.a[\mu].l = y.a[\mu].l = a_{\lambda}[\mu].l$}{
        $\Gamma_0 \leftarrow \{ (a,p) \in \Gamma ~|~ a = 
        \max\left(b | \exists q, (b,q) \in \Gamma, b[\mu].s = a_{\lambda}[\mu].s\right) \} $\;
        \lIf{$\Gamma_0 = \{(a,p)\}$}{
          $p_{\lambda} \leftarrow p$
        }\;
        \lElse{
          $p_{\lambda} \leftarrow p^*$
        }\;
      }
      \lElse{$p_{\lambda} \leftarrow p^*$}\;
    }\lElse{$p_{\lambda} \leftarrow p^*$\label{algline:loop_p2_process_replies_end}}\;
    $\forall \alpha \in \Pi$, \KwSend{$\alpha, \langle p2a, \lambda, a_{\lambda}, p_{\lambda} \rangle$}\label{algline:loop_bcast_p2a}\;
    \lIf{ \KwPR{$2$} returns $nok$}{go to [Ph. 1]}\;
    $\forall \alpha \in \Pi$, \KwSend{$\alpha, \langle decision, \lambda, a_{\lambda}, p_{\lambda} \rangle$}\label{algline:loop_bcast_dec}\;
  }
  \caption{Proposer $\lambda$ - Main loop}
  \label{alg:proposer}
\end{algorithm}

\begin{algorithm}
  \footnotesize
  % \DontPrintSemicolon
  \SetKwBlock{KwLoop}{loop}{end loop}
  \SetKwBlock{KwFunction}{function}{end function}
  \SetKwBlock{KwVariables}{variables}{}
  \SetKwFunction{KwReceive}{receive}  
  \SetKwFunction{KwFillCl}{fill\_cl}
  \SetKwFunction{KwUpdate}{update}
  \SetKwFunction{KwInput}{input}
  \SetKwFunction{KwCheckEntry}{check\_entry}
  \SetKwFunction{KwExec}{execute}
  \SetKwFunction{KwPR}{PR}

  \KwFunction({\KwPR{$\phi$ : phase 1 or phase 2}}){
    $N \leftarrow \emptyset$, $M \leftarrow 0$, $a_{sent} \leftarrow a_{\lambda}$\;
    \While{$|N| < \gconst{n} - \gconst{f}$}{
      $b \leftarrow a_{sent}$\label{algline:pr_begin_loop}\;
      $\langle p\phi b, \alpha, a_{\alpha}, q_{\alpha} \rangle \leftarrow$ \KwReceive{$\langle p\phi b, *, *,* \rangle$}\label{algline:pr_receive}\;
      \KwFillCl($a_{\alpha},b$)\label{algline:pr_alpha_b_fill_cl}\;
      $C^+ = (a_{\alpha} \simeq b) 
      \land (\phi = 2 \Rightarrow p_{\lambda} = q_{\alpha}.p))$\label{algline:pr_cplus}\;
      $C^- = (a_{\alpha} \not\preccurlyeq b)$\label{algline:pr_cminus}\;
      \If{$\alpha \not\in N$}{
        \lIf{$C^+ \lor C^-$}{$N \leftarrow N \cup \{\alpha\}$}\;
        \lIf{$C^+$}{$M \leftarrow M + 1$}\;
        \uElse{
          \lIf{$a_{\alpha}[\lambda].(l$ or $cl) \not\preccurlyeq a_{\lambda}[\lambda].l$}
          {$H^{cl}_{\lambda} \leftarrow H^{cl}_{\lambda} + a_{\alpha}[\lambda].(l$ or $cl)$}\label{algline:pr_update_hcl}\;
          \KwFillCl{$a_{\alpha},a_{\lambda}$}\label{algline:pr_alpha_lambda_fill_cl}\;
          \KwCheckEntry{$\lambda,a_{\lambda},H_{\lambda}^{cl}$}\label{algline:pr_alpha_lambda_check_entry}\;
          let $\mu = \chi(a_{\alpha})$ in \label{algline:pr_negreply_1}\;
          \If{$a_{\alpha} \not\preccurlyeq a_{\lambda}$}{
            \uIf{$\mu < \chi(a_{\lambda})$}{
              $H_{\lambda}[\mu] \leftarrow H_{\lambda}[\mu] + a_{\lambda}[\mu].l$\label{algline:pr_mu_less_lambda1}\;
              $a_{\lambda}[\mu] \leftarrow a_{\alpha}[\mu]$\;
              \lIf{$\exists l \in H_{\lambda}[\mu],~ l \not\preccurlyeq_{l} a_{\lambda}[\mu].l$}{$a_{\lambda}[\mu].cl \leftarrow l$}\;
              $a_{\lambda} \leftarrow \nu^t(\lambda, a_{\lambda}, H_{\lambda}^{cl})$\label{algline:pr_mu_less_lambda2}\;
            }
            \Else{
              (we have $\chi(a_{\lambda}) = \mu$ and $a_{\lambda}[\mu].l = a_{\alpha}[\mu].l$)\label{algline:pr_mu_eq_chi}\;
              \uIf{$a_{\alpha}[\mu].s = a_{\lambda}[\mu].s$}{
                $a_{\lambda}[\mu].t \leftarrow a_{\alpha}[\mu].t$\;
                $a_{\lambda} \leftarrow \nu^t(\lambda, a_{\lambda}, H_{\lambda}^{cl})$\label{algline:pr_trial_inc}\;
              }
              \Else{
                $a_{\lambda}[\mu].s \leftarrow a_{\alpha}[\mu].s$\;
                $a_{\lambda} \leftarrow \nu^s(\lambda, a_{\lambda}, H_{\lambda}^{cl})$\label{algline:pr_negreply_2}\;
              }
            }     
          }
        }
      }
    }
    \lIf{$M = \gconst{n}-\gconst{f}$}{\Return{$ok$}}\label{algline:pr_return_nok}\;
    \lElse{\Return{$nok$}}\label{algline:pr_return_ok}\;
  }
  \caption{Proposer $\lambda$ - Preempting Routine}
  \label{alg:preemptingroutine}
\end{algorithm}

% \subsection{Generalized Paxos}
% \label{subsec:generalizedpaxos}

\begin{algorithm}
  \footnotesize
  % \DontPrintSemicolon
  \SetKwBlock{KwLoop}{loop}{end loop}
  \SetKwBlock{KwFunction}{function}{end function}
  \SetKwBlock{KwVariables}{variables}{}
  \SetKwBlock{KwRule}{rule}{end~rule}
  \SetKwFunction{KwSend}{send}
  \SetKwFunction{KwReceive}{receive}
  \SetKwFunction{KwBcast}{broadcast}
  \SetKwFunction{KwPropagate}{propagate}
  \SetKwFunction{KwDecide}{decide}
  \SetKwFunction{KwFillCl}{fill\_cl}
  \SetKwFunction{KwInput}{input}
  \SetKwFunction{KwTruncate}{truncate}
  \SetKwFunction{KwCheckEntry}{check\_entry}
  \KwFunction(\KwDecide{$b$ : tag, $p$ : command sequence}){
      $learned_{\alpha} \leftarrow p$\;
  }
  \KwFunction(\KwTruncate{$a$ : tag, $p$ : command sequence}){
    \lIf{$|p| > a[\chi(a)].s - 1$}{$p \leftarrow$ the suffix of $p$ of length $a[\chi(a)].s - 1$}\;
    \lIf{$|p| < a[\chi(a)].s - 1$}{append $nop$ to $p$ until $|p| = a[\chi(a)].s - 1$}\;
  }
  \caption{Generalized Paxos - Procedure \texttt{decide} and \texttt{truncate}, acceptor $\alpha$}
  \label{alg:gpaxos_decide}
\end{algorithm}

\begin{algorithm}
  \footnotesize
  % \DontPrintSemicolon
  \SetKwBlock{KwLoop}{loop}{end~loop}
  \SetKwBlock{KwBlock}{}{}
  \SetKwBlock{KwRule}{rule }{end~rule}
  \SetKwBlock{KwFunction}{function}{end~function}
  \SetKwBlock{KwVariables}{variables}{}
  \SetKwFunction{KwSend}{send}
  \SetKwFunction{KwBcast}{broadcast}
  \SetKwFunction{KwFillCl}{fill\_cl}
  \SetKwFunction{KwUpdate}{update}
  \SetKwFunction{KwInput}{input}
  \SetKwFunction{KwTruncate}{truncate}
  \SetKwFunction{KwExec}{execute}
  \SetKwFunction{KwPR}{PR}
  \KwLoop({As long as $\Theta_{\lambda} = \gconst{true}$}){
    $p^* \leftarrow p_{\lambda}$\;
    $cmd \leftarrow \KwInput()$\label{algline:gpaxos_loop_read_proposal}\;
    $a_{\lambda} \leftarrow \nu^s(\lambda,a_{\lambda},H_{\lambda}^{cl})$\;
    \KwTruncate{$a_{\lambda}$, $p^*$}\;
    [Ph. 1]\;
    $p_{\lambda} \leftarrow p^* \circ cmd$\;\label{algline:gpaxos_loop_append1}\;
    $\forall \alpha \in \Pi$, \KwSend{$\alpha, \langle p1a, \lambda, a_{\lambda} \rangle$}\label{algline:gpaxos_loop_bcast_p1a}\;
    \uIf{ \KwPR{$1$} returns $nok$}{
      \uIf{$\chi(a_{\lambda})$ or $a_{\lambda}[\chi(a_{\lambda})].l$ has changed}{ \KwTruncate{$a_{\lambda}$,$p^*$}\;}
      go to [Ph. 1]\;
    }\;
    [Ph. 2]\;
    let $\mu = \chi(a_{\lambda})$, and $\Gamma$ be the set of non-null proposals $r_{\alpha}[\mu]$\label{algline:gpaxos_loop_p2_begin}
    received at the end of [Ph. 1] in\;
    \uIf{$\Gamma \neq \emptyset$}{
      \uIf{$\forall x,y \in \Gamma, \chi(x.a) = \chi(y.a) = \mu \land x.a[\mu].l = y.a[\mu].l = a_{\lambda}[\mu].l$}{
        $\Gamma_0 \leftarrow \{ (a,p) \in \Gamma ~|~ a = 
        \max\left(b | \exists q, (b,q) \in \Gamma, b[\mu].s = a_{\lambda}[\mu].s = |q| \right) \} $\;
        \uIf{$\Gamma_0$ is not empty}{
          let $p_{max}$ be the maximum (lexicographically) command history in $\Gamma_0$\;
          $p_{\lambda} \leftarrow p_{max} \circ cmd$\;\label{algline:gpaxos_loop_append2}\;
        }\lElse{$p_{\lambda} \leftarrow p^* \circ cmd$}\;
      }\lElse{$p_{\lambda} \leftarrow p^* \circ cmd$}\;
    }\lElse{$p_{\lambda} \leftarrow p^* \circ cmd$}\;
    $\forall \alpha \in \Pi$, \KwSend{$\alpha, \langle p2a, \lambda, a_{\lambda}, p_{\lambda} \rangle$}\label{algline:gpaxos_loop_bcast_p2a}\;
    \uIf{ \KwPR{$1$} returns $nok$}{
      \uIf{$\chi(a_{\lambda})$ or $a_{\lambda}[\chi(a_{\lambda})].l$ has changed}{\KwTruncate{$a_{\lambda}$,$p^*$}\;}
      go to [Ph. 1]\;
    }\;
    $\forall \alpha \in \Pi$, \KwSend{$\alpha, \langle decision, \lambda, a_{\lambda}, p_{\lambda} \rangle$}\label{algline:gpaxos_loop_bcast_dec}\;
  }
  \caption{Generalized Paxos -  Main loop, proposer $\lambda$}
  \label{alg:gpaxos_proposer}
\end{algorithm}

\end{document}